\documentclass[letterpaper,11pt]{article}

\usepackage[cmex10]{amsmath}
\usepackage{amsxtra,amssymb,amsthm,amsfonts}
\usepackage{mathtools}
\usepackage[matha,mathx]{mathabx}
\usepackage{mathrsfs}
\usepackage[usenames]{color}
\usepackage{umoline}
\usepackage{datetime}
\usepackage[margin=1in]{geometry}
\usepackage[ruled,boxed,linesnumbered]{algorithm2e}
\usepackage{setspace}
\usepackage{bbm}

\usepackage{enumerate}
\usepackage{graphicx}
\usepackage{caption}
\usepackage{subfig}
\usepackage{amssymb}
\usepackage{array}

\usepackage{tikz}
\usepackage{pgfplots}
\newtheorem{lemma}{Lemma}[section]
\newtheorem{prop}[lemma]{Proposition}
\newtheorem{theorem}[lemma]{Theorem}

\newtheorem{rem}[lemma]{Remark}

\newcommand{\re}{\begin{rem}\rm}
\newcommand{\mar}{\end{rem}}

\newtheorem{defi}[lemma]{Definition}

\newcommand{\fo}{\begin{eqnarray*}}
\newcommand{\mel}{\end{eqnarray*}}

\newcommand{\cz}{{\mathbb C}}

\newcommand{\id}{\mathrm{id}}

\newcommand{\pl}{\hspace{.1cm}}

\newcommand{\A}{{\mathcal A}}

\newcommand{\C}{{\mathcal C}}

\newcommand{\calH}{{\mathcal H}}

\newcommand{\calM}{{\mathcal M}}

\newcommand{\calS}{{\mathcal S}}
\newcommand{\calT}{{\mathcal T}}

\newcommand{\qd}{\end{proof}\vspace{0.5ex}}

\newcommand\hsnorm[1]{\norm{#1}_{\mathrm{HS}}}

\makeatletter
\newcommand{\opnorm}{\@ifstar\@opnorms\@opnorm}
\newcommand{\@opnorms}[1]{%
  \left|\mkern-1.5mu\left|\mkern-1.5mu\left|
   #1
  \right|\mkern-1.5mu\right|\mkern-1.5mu\right|
}
\newcommand{\@opnorm}[2][]{%
  \mathopen{#1|\mkern-1.5mu#1|\mkern-1.5mu#1|}
  #2
  \mathclose{#1|\mkern-1.5mu#1|\mkern-1.5mu#1|}
}
\makeatother

\newcommand{\norm}[1]{\Vert#1\Vert}

\newcommand{\trace}{\mathrm{tr}}
\newcommand{\transpose}{\top}
\newcommand{\diag}{\mathrm{diag}}

\newcommand{\conv}{\circledast}
\newcommand{\supp}[1]{\mathrm{supp}(#1)}

\newcommand{\fnorm}[1]{\|#1\|_{\rm F}}

\newcommand{\argmin}{\mathop{\rm argmin}}
\newcommand{\argmax}{\mathop{\rm argmax}}

\title{Blind Recovery of Sparse Signals from Subsampled Convolution}
\author{Kiryung Lee, Yanjun Li, Marius Junge, and Yoram Bresler}

\begin{document}
\doublespacing

\maketitle

%\tableofcontents\pagebreak

\begin{abstract}
  Subsampled blind deconvolution is the recovery of two unknown signals from samples of their convolution. To overcome the ill-posedness of this problem, solutions based on priors tailored to specific application have been developed in practical applications. In particular, sparsity models have provided promising priors. However, in spite of empirical success of these methods in many applications, existing analyses are rather limited in two main ways: by disparity between the theoretical assumptions on the signal and/or measurement model versus practical setups; or by failure to provide a performance guarantee for parameter values within the optimal regime defined by the information theoretic limits. In particular, it has been shown that a naive sparsity model is not a strong enough prior for identifiability in the blind deconvolution problem. Instead, in addition to sparsity, we adopt a conic constraint, which enforces spectral flatness of the signals. Under this prior, we provide an iterative algorithm that achieves guaranteed performance in blind deconvolution at near optimal sample complexity. Numerical results show the empirical performance of the iterative algorithm agrees with the performance guarantee.
\end{abstract}

\section{Introduction}

Blind deconvolution is the estimation of two signals from their convolution with one another without knowing either signal.
The blind deconvolution problem arises in numerous applications including
astronomical speckle imaging, fluorescence microscopy, remote sensing, wireless communications, seismic data analysis,
speech dereverberation, and medical imaging (cf. \cite{kundar1996blind}).
Without further information, the blind deconvolution problem does not admit a unique solution.
However, prior information on signals of interest in practical applications enabled resolution of unknown signals from their convolution.
Most approaches formulated the recovery as a regularized or constrained nonlinear least squares problem using prior information.
In particular, subspace and sparsity models have been employed as promising priors and enabled empirical success in many practical applications.
For example, a finite impulse response (FIR) model corresponds to a subspace model where the subspace is spanned by the standard basis vectors.
Deterministic sparsity models have been employed as a signal prior in applications such as echo cancellation \cite{liu2005relevant} and seismic data analysis \cite{kazemi2014sparse}.
Statistical models with heavy-tailed distributions (e.g., \cite{kotera2013blind}) also promote sparsity of the solution.

In certain applications, the recovery from sub-sampled convolution is needed.
FIR or more general sparsity priors enabled the recovery from a small number of samples of the convolution in superresolution \cite{sroubek2007unified} and in parallel MRI \cite{ying2007joint}.
Such applications motivate our inclusion of the subsampled convolution case in this paper.
As we will see, our analysis addresses this case with little or no extra effort.

In wireless communications and image processing, the blind deconvolution of a single input signal from multiple channel outputs has been of interest. Under the assumption that the unknown channels responses can be represented as short FIR filters,
subspace methods and their algebraic performance guarantees have been studied \cite{abed1997blind}.
Unfortunately, neither algorithms nor theory directly apply to problems with other priors such as sparsity or nonnegativity,
in particular when there is only a single channel output.
In these more challenging scenarios, commonly used approaches have been alternating regularized least squares,
where regularizer terms are tailored to priors given in specific applications.
Although these methods were empirically successful, a rigorous performance guarantee has been missing or rather limited to date.

By lifting the reconstruction problem to a higher dimensional problem of recovering a rank-1 matrix,
Ahmed et al. \cite{ahmed2014blind} proposed a convex optimization approach to blind deconvolution.
They proved the following near optimal performance guarantee.
Under certain subspace models on the signals and assumptions of randomness,
with the subspace dimensions proportional (up to a logarithmic factor) to the signal length,
the recovery of the signals is guaranteed with high probability.
This performance guarantee has near optimal scaling in the information theory sense: the number of measurements required (equal to the signal length) is proportional up to a log factor to the number of degrees of freedom in the signals.
However, the subspace model is not general enough to describe signal priors employed in most practical applications.

Since the first result on a near optimal performance guarantee for blind deconvolution using subspace models \cite{ahmed2014blind},
a few subsequent works extended the result in various ways, in particular, in terms of generalizing the signal prior
from a subspace model to a sparsity model corresponding to a union-of-subspaces.
Ling and Strohmer \cite{ling2015self} proposed to use a convex relaxation approach that drops the rank-1 constraint in the lifted formulation.
Their performance guarantee allows the sparsity level to be almost proportional (up to a logarithmic factor) to the square root of the signal length, which is suboptimal scaling compared to the number of degrees of freedom in the signal model.
Chi \cite{chi2015guaranteed} studied a performance guarantee for a more challenging infinite dimensional blind deconvolution problem,
referred to as blind spike deconvolution, where one signal is spikes with continuous-valued shifts and the other signal belongs to a random subspace.
Her solution too employed a convex relaxation based on an atomic model and the corresponding guarantee was given at a sample complexity proportional to the product of the square of the sparsity level and the dimension of the subspace.

Motivated by blind deconvolution with sparsity priors and by other bilinear inverse problems sharing a similar structure,
some of the authors of this paper together with Yihong Wu proposed an iterative algorithm
called the \textit{sparse power factorization} (SPF) \cite{LeeWB2013spf}.
Similar to the aforementioned two theoretical works \cite{ahmed2014blind,choudhary2014identifiability},
a general bilinear inverse problem was lifted to the problem of recovering a simultaneously sparse and rank-1 matrix.
In a prototypical setup, where the linear measurements are obtained as inner products with i.i.d. Gaussian matrices,
a near optimal performance guarantee for SPF was shown, which achieves the information theoretic fundamental limit.
Furthermore, empirically, SPF outperformed combinations of convex relaxations of both low-rankness and sparsity priors.
The performance guarantees for SPF were derived based on the restricted isometry property of an i.i.d. Gaussian measurement operator.
Unfortunately, the linear operator arising in blind deconvolution does not satisfy such a strong property;
hence, the near optimal performance guarantees for SPF \cite{LeeWB2013spf} do not apply to the blind deconvolution problem straightforwardly.

In this paper, as a sequel to the previous work \cite{LeeWB2013spf},
we propose an alternating minimization algorithm modified from SPF and provide its performance guarantee for the blind deconvolution problem.
The linear measurements in the blind deconvolution problem are obtained by inner products with rank-1 matrices.
Similar to the matrix completion problem \cite{candes2010power}, some notion of incoherence is necessary:
in fact, a naive sparsity prior has been shown not strong enough for identifiability in the blind deconvolution problem \cite{choudhary2014sparse}. In addition to sparsity, we adopt a prior on the signal proposed by Ahmed et al. \cite{ahmed2014blind}
expressing a preference for flat spectra in the Fourier domain.
In this setup, when the unknown sparse signals are heavily peaked, i.e.,
a few dominant components contain a certain fixed fraction of the total signal energy,
the proposed alternating minimization algorithm provides stable recovery for signals of sparsity levels
almost proportional (up to a logarithmic factor) to the signal length.
This is a performance guarantee at near optimal sample complexity and improves on the known theoretical results in blind deconvolution significantly. (The comparison is presented in Table~\ref{tab:comparison}.)
The proofs of our theoretical results are based on a new RIP-like property derived in a companion paper \cite{LeeJunge2015}.
Unlike the aforementioned relevant works, our performance guarantee, derived under the additional flat spectrum property,
provides blind deconvolution at near optimal scaling of the sample complexity.
Furthermore, our work is the first to achieve a near optimal performance guarantee
for the significantly more challenging recovery problem of blind deconvolution from subsampled data.

\begin{table}
  \begin{center}
  \begin{tabular}{c|c|c|c|c}
  & Signal Model & \#Measurements & Subsampling & Additional Assumptions \\\hline\hline
  \cite{ahmed2014blind} & subspace & $O(s \log n)$ & No & spectral flatness \\\hline
  \cite{ling2015self} & union of subspaces & $O(s^2 \log n)$ & No & \\\hline
  \cite{chi2015guaranteed} & mixed (off-grid) & $O(s^3 \log n)$ & No & separation of frequencies \\\hline
  This work & union of subspaces & $O(s \log n)$ & $\checkmark$ & spectral flatness \\\hline
  \end{tabular}
  \end{center}
  \caption{Comparison of the performance guarantees for blind deconvolution.}
  \label{tab:comparison}
\end{table}

Another line of research \cite{choudhary2014identifiability,li2015unified,li2015identifiability} on blind deconvolution
has pursued algebraic identifiability results.
Chowdhary et al. \cite{choudhary2014identifiability} explored algebraic conditions that enable unique identification of unknown signals in blind deconvolution (and in its generalization to arbitrary bilinear inverse problems) in a deterministic setup.
However, their analysis, even in a noise-free scenario without any restriction on computational cost,
did not provide a condition for recovery explicitly given in terms of sample complexity.
Unlike the earlier algebraic analysis, recent results by Li et al. \cite{li2015unified,li2015identifiability}
provide performance guarantees given explicitly in terms of sample complexity.
While these results are restricted to the noiseless scenario and do not provide a constructive algorithm at a polynomial computational cost,
they provide sharp analysis on the relation among model parameters without involving conservative absolute constants.
Most notably, a very recent result by Li et al. \cite{li2015identifiability} provides a tight analysis by
closing the gap between necessary and sufficient conditions in previous works.
Unlike the result in this paper, the extension to the subsampled case remains an open problem.

The rest of this paper is organized as follows.
The subsampled blind deconvolution problem is formulated with signal priors in Section~\ref{sec:formulation}
and a provably near optimal iterative algorithm and its practical implementation are discussed in Section~\ref{sec:algorithm}.
We present performance guarantees via RIP-like properties and their implications for sample complexity in Section~\ref{sec:main}.
The proofs of the main results are presented in Section~\ref{sec:proofs}.
After a discussion of the numerical results in Section~\ref{sec:numres},
we conclude the paper with a summary in Section~\ref{sec:concl}.

\section{Problem Statement}
\label{sec:formulation}

\subsection{Notation}

Various norms are used in this paper.
The Frobenius norm of a matrix is denoted by $\norm{\cdot}_{\mathrm{F}}$.
The operator norm from $\ell_p^n$ to $\ell_q^n$ will be $\norm{\cdot}_{p\to q}$.
Absolute constants will be used throughout the paper.
Symbols $C,c_1,c_2,\ldots$ are reserved for real-valued positive absolute constants.
Symbol $\beta$ is a positive integer absolute constant.
For a matrix $A$, its element-wise complex conjugate, its transpose, and its Hermitian transpose
are respectively written as $\overline{A}$, $A^\transpose$, and $A^*$.
For a linear operator $\A$ between two vector spaces, $\A^*$ will denote its adjoint operator.
The matrix inner product $\trace(A^*B)$ between two matrices $A$ and $B$ is denoted by $\langle A, B \rangle$.
Matrix $F \in \cz^{n \times n}$ will be used to denote the unitary discrete Fourier transform, and $\conv$ stands for the circular convolution.
We will use the shorthand notation $[n] = \{1,2, \ldots, n\}$.
Let $J \subset [n]$. Then, $\Pi_J: \cz^n \to \cz^n$ denotes the coordinate projection
whose action on a vector $x$ keeps the entries of $x$ indexed by $J$ and sets the remaining entries to zero.
The identity map on $\cz^{n \times n}$ will be denoted by $\id$.
Finally, we use $\Gamma_s$ to denote the set of $s$-sparse vectors in $\cz^n$, i.e.,
\begin{align*}
\Gamma_s := \{u \in \cz^n:\pl \norm{u}_0 \leq s\},
\end{align*}
where $\norm{u}_0$ counts the number of nonzero elements in $u$.

\subsection{Measurement model}

The subsampled blind deconvolution problem is formulated as a bilinear inverse problem as follows.
Let $\Omega = \{\omega_1,\omega_2,\ldots,\omega_m\} \subset [n]$ denote the ordered set of $m$ distinct sampling indices.
Given $\Omega$, the sampling operator $S_\Omega: \cz^n \to \cz^m$ is defined so that
the $k$th element of $S_\Omega x \in \cz^m$ is the $\omega_k$th element of $x \in \cz^n$ for $k = 1,\ldots,m$.
Then, the $m$ samples of the convolution $x \conv y$ indexed by $\Omega$ corrupted by additive noise $z$ constitute
the measurement vector $b \in \cz^m$, which is expressed as
\begin{equation}
\label{eq:mdl_meas}
b = \sqrt{\frac{n}{m}} S_\Omega (x \conv y) + z.
\end{equation}

Let $x,y \in \cz^n$ be uniquely represented as $x = \Phi u$ and $y = \Psi v$ over dictionaries $\Phi$ and $\Psi$.
Then, the recovery of $(x,y)$ is equivalent to the recovery of $(u,v)$.
Since each element of $b$ corresponds to a bilinear measurement of $(u,v)$,
the subsampled blind deconvolution problem corresponds to the bilinear inverse problem of recovering $(u,v)$
from its bilinear measurements in $b$, when $\Omega$, $\Phi$, and $\Psi$ are known.

Ahmed et al. \cite{ahmed2014blind} proposed to solve the blind deconvolution problem as recovery of a rank-1 matrix from its linear measurements.
By the lifting procedure, bilinear measurements of $(u,v)$ are equivalently rewritten as linear measurements of the matrix $X = u v^\transpose$, i.e., there is a linear operator $\A: \cz^{n \times n} \to \cz^m$ such that
\begin{equation}
\label{eq:defcalA}
b = \A(X) + z \pl.
\end{equation}
Then, each element of the measurement vector $b$ corresponds to a matrix inner product.
Indeed, there exist matrices $M_1,M_2,\ldots,M_m \in \cz^{n \times n}$ that describe the action of $\A$ on $X$ by
\begin{equation}
\label{eq:defcalAbyMell1}
\A(X) = [\langle M_1, X \rangle, \ldots, \langle M_m, X \rangle]^\transpose.
\end{equation}
Since the circular convolution corresponds to the element-wise product in the Fourier domain, the matrices are explicitly expressed as
\begin{equation}
\label{eq:defcalAbyMell2}
M_\ell = \sqrt{n} \Phi^* F^* \diag(f_\ell) \overline{F} \overline{\Psi}, \quad \ell = 1,\ldots,m,
\end{equation}
where $f_\ell$ denotes the $\ell$th column of the unitary DFT matrix $F \in cz^{n \times n}$.
The blind deconvolution problem then becomes a linear inverse problem
with a matrix-valued unknown variable $X$ that is additionally constrained to the set of rank-1 matrices.

In the lifted formulation, a reconstruction $\widehat{X}$ of the unknown matrix $X$ is considered successful
if it satisfies the following stability criterion:
\begin{equation}
\label{eq:success}
\frac{\norm{\widehat{X} - X}_{\mathrm{F}}}{\norm{X}_{\mathrm{F}}}
\leq C \left( \frac{\norm{z}_2}{\norm{\A(X)}_2} \right)
\end{equation}
for an absolute constant $C$.
This is a natural criterion, requiring the relative reconstruction error to be at most proportional to the signal to noise ratio in the measurements. Indeed, up to the particular value of $C$, this is the best one could hope for in the linear inverse problem (\ref{eq:mdl_meas})
even if the system of linear equation were fully determined.
Furthermore, this definition of success is free of the inherent scale ambiguity
($x \conv y =(\alpha^{-1} x) \conv (\alpha y), ~ \forall \alpha \neq 0$) in the original bilinear formulation.
The shift ambiguity coming from the circular convolution is removed when
the dictionaries $\Phi$ and $\Psi$ are not shift-invariant,
i.e., not all circular shifts of the atoms in $\Phi$ and $\Psi$ are also included in $\Phi$ and $\Psi$, respectively.
Once $\widehat{X}$ is recovered, $u$ (resp. $v$) is identified up to a scale factor as the left (resp. right) factor of the rank-1 matrix $\widehat{X}$.

\subsection{Priors on signals}
\label{subsec:priors}

The subsampled blind deconvolution problem does not admit a unique solution and cannot be solved
without placing some restrictions on the unknown signals.
In this paper, we use priors based on sparsity models to solve the sub-sampled blind deconvolution problem.

First, we assume that signals $x$ and $y$ are sparse over $\Phi$ and $\Psi$, respectively.
In other words, the coefficient vectors $u$ and $v$ are sparse at sparsity levels $s_1$ and $s_2$, respectively,
which we denote by $u \in \Gamma_{s_1}$ and $v \in \Gamma_{s_2}$, where $\Gamma_s$ is the set of $s$-sparse vectors in $\cz^n$.
Geometrically, $u$ (resp. $v$) belongs to the union of all subspaces spanned by $s_1$ (resp. $s_2$) standard basis vectors.
Equivalently, the signal $x$ belongs to a union of subspaces,
each spanned by $s_1$ columns of $\Phi$ (with an analogous statement for $y$).
From this perspective, the subspace model  considered by previous authors \cite{ahmed2014blind} corresponds to the special case
where the particular subspace in the union to which $u$ (resp. $v$) belongs is known a priori.

A union of subspaces model has proved to be an effective prior for various ill-posed linear inverse problems,
including, most notably, reconstruction in compressed sensing.
However, blind deconvolution is a more challenging ill-posed bilinear inverse problem and
it has been shown \cite{choudhary2014sparse} that sparsity alone does not provide a prior strong enough for stable recovery in blind deconvolution.
Therefore, we augment the sparsity prior by an additional prior called ``spectral flatness'' \cite{ahmed2014blind}.

For a signal $x \in \cz^n$, its spectral flatness level is defined by
\[
\textsf{sf}(x) := \frac{n \norm{F x}_\infty^2}{\norm{F x}_2^2}.
\]
Clearly, $\texttt{sf}(x) \geq 1$ for any $x \in \cz^n$, with equality achieved by a signal with a perfectly flat spectrum.
The set $C_\mu$ of signals at spectral flatness level $\mu$ is then
\begin{equation}
\label{eq:defCmu}
\C_{\mu} := \{x \in \cz^n:\pl \textsf{sf}(x) \leq \mu \}.
\end{equation}
Combining these constraints with the sparsity constraints over dictionaries $\Phi$ and $\Psi$, respectively,
we assume that $x \in \Phi \Gamma_{s_1} \cap \C_{\mu_1}$ and $y \in \Psi \Gamma_{s_2} \cap \C_{\mu_2}$.
When $\Phi$ and $\Psi$ are invertible\footnote{For simplicity, we restrict our analysis to the case where $\Phi$ and $\Psi$ are invertible matrices. However, it is straightforward to extend the analysis to the case with overcomplete dictionaries by replacing the inverse by the preimage operator.}, this is equivalent to assuming that $u \in \Gamma_{s_1} \cap \Phi^{-1} \C_{\mu_1}$ and $v \in \Gamma_{s_2} \cap \Psi^{-1} \C_{\mu_2}$.
In fact, a ``flat-spectrum'' property with $\textsf{sf}(y) = O(\log n)$ was crucial
in deriving a near optimal performance guarantee under the subspace model \cite{ahmed2014blind}.

The non-convex cone $\C_{\mu}$ does not share some of the useful properties satisfied by a union of subspaces $\Gamma_s$.
For example, whereas $\Gamma_s$ satisfies $\Gamma_s + \Gamma_s \subset \Gamma_{2s}$,
$\C_{\mu}$ does not satisfy the analogous property $\C_{\mu} + \C_{\mu} \subset \C_{2\mu}$.
Technically, as we will see, the lack of this property requires new RIP-like properties to derive our guarantees.
We derive such RIP-like properties in a companion paper \cite{LeeJunge2015}.

Note that the equivalence between the bilinear formulation and the lifted rank-constrained linear formulation remains valid
in the presence of the additional constraints corresponding to the signal priors.
In the lifted formulation, $X$ is factorized as $X = u v^\transpose$ subject to the constraints
that $u \in \Gamma_{s_1} \cap \Phi^{-1} C_{\mu_1}$ and $v \in \Gamma_{s_2} \cap \Psi^{-1} \C_{\mu_2}$.
Although factorizing a rank-1 matrix as the outer product of two vectors involves a scale ambiguity, this has no effect on the constraints: 
since $\Gamma_{s_1} \cap \Phi^{-1} C_{\mu_1}$ and $\Gamma_{s_2} \cap \Psi^{-1} \C_{\mu_2}$ are closed under scalar multiplication,
the constraints due to signal priors in the lifted formulation coincide with those assumed for the bilinear formulation.

Our performance guarantees in this paper apply to signals in the intersection $\Gamma_s \cap \Phi^{-1} \C_{\mu}$.
While the sparsity prior is well accepted as a signal model in applications including blind deconvolution,
the spectral flatness prior is a relatively new concept, introduced only recently \cite{ahmed2014blind}.
Therefore, it is of interest to understand the relation between the two priors, or equivalently,
the constraint sets $\Gamma_s$ and $\Phi^{-1} \C_{\mu}$.
In the remainder of this section,
we discuss how the two constraint sets, respectively corresponding to the sparsity and spectral flatness priors, are related.
In particular, we analyze this relationship in the setup
where $\Phi$ is a random matrix whose entries are i.i.d. with the circular complex normal distribution $CN(0,1/n)$.
(Later, we will derive performance guarantees with optimal scaling on parameters in this setup.)

We consider the following three questions:
\begin{description}
  \item[(i)] When is $\Gamma_s \subset \Phi^{-1} \C_\mu$ -- that is, when does the sparsity condition subsume the spectral flatness condition?
  \item[(ii)] When is $\Phi^{-1} \C_\mu \cap \Gamma_s$ a proper subset of $\Gamma_s$ -- that is, when is the spectral flatness constraint active?
  \item[(iii)] For the choice $\mu = O(\log n)$ (independent of $s$), which is favorable for our performance guarantee, does the intersection  $\Phi^{-1} \C_\mu \cap \Gamma_s$ contain signals of sparsity close to $s$ -- that is, is the combined prior not too restrictive?
\end{description}

The answer to the first question is provided by the following proposition, proved in Appendix~\ref{appendix:prop2.1}.

\begin{prop}
\label{Proposition 2.1}
Let $\Phi$ be a random matrix whose entries are i.i.d. following $CN(0,1/n)$.
Let $\eta \in (0,1)$. There exists an absolute constant $C$ for which the following holds. 
Suppose
\begin{equation}
\label{A1}
s < \eta^2 C n / \log(n/s).
\end{equation}
Then, with high probability, the worst case spectral flatness level for $u \in \Gamma_s$ is upper-bounded as
\begin{equation}
\label{eq:prop2.1}
\sup_{u \in \Gamma_{s}} \textsf{\upshape sf}(\Phi u) \leq \frac{c s \log n}{1-\eta}.
\end{equation}
\end{prop}

By Proposition~\ref{Proposition 2.1}, for $u \in \Gamma_s$, subject to (\ref{A1}),
we have with high probability that $\textsf{sf}(\Phi u) = O(\norm{u}_0 \log n)$.
It follows that there exists an absolute constant $c$ such that subject to (\ref{A1}),
if $\mu > c s \log n$, then with high probability, $\Gamma_s \subset \Phi^{-1} \C_\mu$,
and the spectral flatness constraint is not active.

The answer to Question (ii) is provided by the following proposition, proved in Appendix~\ref{appendix:prop2.2}.
\begin{prop}
\label{Proposition 2.2}
Let $\Phi$ be a random matrix whose entries are i.i.d. following $CN(0,1/n)$.
Let $\eta \in (0,1)$. There exists an absolute constant $C$ for which the following holds. 
Suppose (\ref{A1}) holds (for given $\eta$ and $C$).
Then, with nonzero probability (approximately 0.5), there exists an $s$-sparse $u$ such that
\begin{equation}
\label{eq:prop2.2}
\textsf{\upshape sf}(\Phi u) \geq \frac{s[1-2/(9s)]^3}{1+\eta}.
\end{equation}
\end{prop}

It follows that if $\mu$ is sublinear in $s$, then $\Gamma_s$ is not necessarily contained in $\Phi^{-1} \C_{\mu}$, 
and the spectral flatness constraint may be active.

As an aside, Proposition~\ref{Proposition 2.2} demonstrates that
the (probabilistic) upper bound on Proposition~\ref{Proposition 2.1} on the worst case spectral flatness
for $u \in \Gamma_s$ is nearly tight (up to a logarithmic factor).

Our Question (iii) above is motivated by the following considerations.
Our performance guarantees in Section~\ref{sec:main} will hold with sample complexity of $m = O(s \mu \log^\beta n)$.
Near optimal scaling of parameters is achieved when $\mu = O(\log n)$.
In this case, the performance guarantees do not apply to all $s$-sparse signals
but rather to those $s$-sparse signals whose spectral flatness satisfies $\mu = O(\log n)$.

One may wonder whether posing both priors simultaneously might be too restrictive, in particular at spectral flatness level $\mu = O(\log n)$.
More specifically, one may suspect that only very sparse signals are in the intersection $\Gamma_s \cap \Phi^{-1} \C_{\mu}$.
We show that in fact, this is not the case and there still exist many $s$-sparse signals at spectral flatness level $\mu = O(\log n)$.
The following proposition, proved in Appendix~\ref{appendix:prop2.3}, answers Question (iii) in the positive.

\begin{prop}
\label{Proposition 2.3}
Let $\Phi$ be a random matrix whose entries are i.i.d. following $CN(0,1/n)$.
Let $u$ be fixed arbitrarily in $\Gamma_s$.
Then, with high probability,
\[
\textsf{\upshape sf}(\Phi u) \leq c \log n
\]
for an absolute constant $c$.
\end{prop}

\section{Blind Deconvolution by Alternating Minimization}
\label{sec:algorithm}

\subsection{Sparse power factorization with projection onto spectrally flat signals}
\label{subsec:altmin}

We propose an alternating-minimization algorithm for subsampled blind deconvolution
that iteratively updates estimates of $u$ and $v$ by alternating between two subproblems,
each defined by fixing an estimate of the other signal.

Recall that the linear operator $\A$ is represented by $m$ measurement functionals defined by the matrices $(M_\ell)_{\ell=1}^m$.
When $v$ is fixed, the action of $\A$ on $u v^\transpose$ reduces to a linear operation on $u$.
The corresponding linear operator has a matrix representation, denoted by $\A_\mathrm{R}(v) \in \cz^{n \times n}$,
which is a linear function of $v$ and expressed as
\begin{equation}
\label{eq:defAR}
\A_\mathrm{R}(v) :=
[\overline{M_1} v, \overline{M_2} v, \ldots, \overline{M_m} v]^\transpose.
\end{equation}
Similarly, when $u$ is fixed, the action of $\A$ on $u v^\transpose$ reduces to a linear operation on $v$.
The corresponding linear operator is represented by a matrix-valued linear function of $u$ given by
\begin{equation}
\label{eq:defAL}
\A_\mathrm{L}(u) :=
[M_1^* u, M_2^* u, \ldots, M_m^* u]^\transpose.
\end{equation}
By the definitions of $\A_\mathrm{R}(v)$ and $\A_\mathrm{L}(u)$, it follows that
\[
\A(u v^\transpose) = [\A_\mathrm{R}(v)] u = [\A_\mathrm{L}(u)] v.
\]
The blind deconvolution algorithm is then described as follows.

Given an estimate $\hat{v}$ of $v$, the estimation of $u$ is cast as a constrained linear inverse problem,
where $\A_\mathrm{R}(\hat{v})$ defines the forward system and
the constraint set is given as the intersection of $\Gamma_{s_1}$ and $\Phi^{-1} \C_{\mu_1}$.
We propose to solve the constrained inverse problem in the following two steps:
first, compute a constrained least squares solution $\tilde{u}$ only with the sparsity constraint;
then, project $\tilde{u}$ to the constraint set $\Gamma_{s_1} \cap \Phi^{-1} \C_{\mu_1}$,
to produce $\hat{u}$. The corresponding projection operator is denoted by $\Phi^{-1} P_{\C_{\mu_1} \cap \Phi \Gamma_{s_1}} (\Phi \tilde{u}_t)$.

The first step, of computing $\tilde{u}$, corresponds to a standard sparse recovery problem,
which can be readily solved by existing algorithms with an RIP-based performance guarantee\footnote{While the sparse recovery algorithms used are standard, their analysis in this application is not. The analysis of sparse recover in Section~\ref{subsec:altminbd_riplike} is fairly elaborate, using RIP-like properties. This is needed, because the error term in the measurement, which is due to the estimation error in the previous step, has  characteristics different to adversary or random noise.}
(e.g., hard-thresholding pursuit (HTP) \cite{Fou2011htp}, compressive sensing matching pursuit (CoSaMP) \cite{needell2009cosamp}, and subspace pursuit \cite{dai2009subspace}).
To present a concrete performance guarantee, we employ for this step the HTP algorithm, summarized in Algorithm~\ref{alg:htp}.\footnote{We stop Algorithm~\ref{alg:htp} if the relative change between consecutive iterates ($\norm{\hat{x}_{t+1} - \hat{x}_t}_2/\norm{\hat{x}_t}_2$) is less than a certain threshold. Similar stopping conditions are used for Algorithms~\ref{alg:altminprojbd} and \ref{alg:approx_proj}.}
Other choices will provide the same performance guarantee (for this step and for the resulting entire blind deconvolution algorithm).
The previous estimate $\hat{v}$ is normalized before the application of the HTP algorithm.
In this way, HTP with the step size set to 1 provides a performance guarantee.

The estimation of $v$ for a given estimate $\hat{u}$ of $u$
is cast as a very similar constrained linear inverse problem.
The forward system is defined by $\A_\mathrm{L}(\hat{u})$ and the constraint set is given as $\Gamma_{s_2} \cap \Phi^{-1} \C_{\mu_2}$.
The same strategy is employed to solve this subproblem.
The entire blind deconvolution algorithm that alternates between the two subproblems is summarized below as Algorithm~\ref{alg:altminprojbd}.

Alternating minimization is a popular heuristic for bilinear inverse problems,
but it may get stuck at local minima. Therefore, it is crucial to find a good initialization.
We propose a simple and guaranteed initialization scheme summarized in Algorithm~\ref{alg:thres_proj}
(which is a modified form of the thresholding initialization of SPF \cite{LeeWB2013spf}).
Algorithm~\ref{alg:thres_proj} aims to find a good estimate $v_0$ of $v$
so that the angle between $v_0$ and $v$ is small and $v_0$ belongs to the constraint set $\Gamma_{s_2} \cap \Psi^{-1} \C_{\mu_2}$.
Algorithm~\ref{alg:thres_proj} first computes estimates $\widehat{J}_1$ and $\widehat{J}_2$ respectively for the supports of $u$ and $v$.
Then, it computes the initialization $v_0$ at sparsity level $s_0$ through a (truncated) singular value decomposition of $\Pi_{\widehat{J}_1} \A^*(b) \Pi_{\widehat{J}_2}$. The size of $\widehat{J}_1$ is set to the known sparsity level $s_1$ of $u$.
The size $s_0$ of $\widehat{J}_2$ is set so that it satisfies the two conditions:
$s_0 \leq s_2$, where $s_2$ is the known sparsity level of $v$, and
\begin{equation}
\label{eq:maxs0}
\frac{\norm{F \Psi_J}_{1 \to \infty}}{\sigma_{\min}(F \Psi_J)} \sqrt{s_0} \leq \sqrt{\frac{\mu_2}{n}}
\end{equation}
for $J = \widehat{J}_2$.

If $\mu_2$ is large enough (an explicit condition is discussed in Section~\ref{subsec:altminbd_riplike}), 
Algorithm~\ref{alg:thres_proj} will terminate to satisfy $s_0 \leq s_2$ and (\ref{eq:maxs0}) for $J = \hat{J}_2$. 
At that point, the resulting $s_0$-sparse $v_0$, which is supported on $\widehat{J}_2$, belongs to $\Gamma_{s_2} \cap \Psi^{-1} \C_{\mu_2}$.
This is readily verified as follows:
Note that $\widehat{J}_2$ denotes the support of $v_0$.
Since
\begin{align*}
\norm{F \Psi v_0}_\infty
\leq \norm{F \Psi_{\widehat{J}_2}}_{1 \to \infty} \norm{v_0}_1
\leq \norm{F \Psi_{\widehat{J}_2}}_{1 \to \infty} \sqrt{s_0} \norm{v_0}_2,
\end{align*}
and
\begin{align*}
\norm{F \Psi v_0}_2
{} & \geq \sigma_{\min}(F \Psi_{\widehat{J}_2}) \norm{v_0}_2,
\end{align*}
it follows that (\ref{eq:maxs0}) for $J = \widehat{J}_2$ implies $v_0 \in \Psi^{-1} \C_{\mu_2}$.
The other containment $v_0 \in \Gamma_{s_2}$ trivially holds since $s_0 \leq s_2$.

Furthermore, $v_0$ is not only feasible but also a ``good'' initialization (close to $v$ in the appropriate sense)
because the rank-1 approximation of $\Pi_{\widehat{J}_1} M \Pi_{\widehat{J}_2}$ is close to $u v^\transpose$
when the RIP-like properties in Section~\ref{sec:main} hold.

\begin{algorithm}
%\setstretch{1.5}
\LinesNumbered
\SetAlgoNoLine
\DontPrintSemicolon
\caption{$v_0 = \texttt{thres\_init}(M,\Psi,s_1,s_2,\mu_2)$}
\label{alg:thres_proj}
$s_0 \leftarrow 1$\;
$\widehat{J}_1 \leftarrow \text{indices of the $s_1$ rows of $M$ with the largest $\ell_\infty$ norm}$\;
$\widehat{J}_2 \leftarrow \text{index of the column of $\Pi_{\widehat{J}_1} M$ with the largest $\ell_2$ norm}$\;
\While{$s_0 \leq s_2$ \& $\norm{F \Psi_{\widehat{J}_2}}_{1 \to \infty} \sqrt{s_0} \leq \sqrt{\mu_2/n} \sigma_{\min}(F \Psi_{\widehat{J}_2})$}{
$s_0 \leftarrow s_0 + 1$\;
\For{k=1,\ldots,n}{
$\zeta_k \leftarrow \text{$\ell_2$ norm of the $s_0$-sparse approx. of the $k$th row of $M$}$\;
}
$\widehat{J}_1 \leftarrow \text{indices of the $s_1$ entries of $\zeta$ with the largest magnitude}$\;
$\widehat{J}_2 \leftarrow \text{indices of the $s_0$ columns of $\Pi_{\widehat{J}_1} M$ with the largest $\ell_2$ norm}$\;
}
$s_0 \leftarrow \max(1,s_0-1)$\;
\For{k=1,\ldots,n}{
$\zeta_k \leftarrow \text{$\ell_2$ norm of the $s_0$-sparse approx. of the $k$th row of $M$}$\;
}
$\widehat{J}_1 \leftarrow \text{indices of the $s_1$ entries of $\zeta$ with the largest magnitude}$\;
$\widehat{J}_2 \leftarrow \text{indices of the $s_0$ columns of $\Pi_{\widehat{J}_1} M$ with the largest $\ell_2$ norm}$\;
$v_0 \leftarrow \text{the first right singular vector of $\Pi_{\widehat{J}_1} M \Pi_{\widehat{J}_2}$}$\;
\Return $v_0$\;
\end{algorithm}

\begin{algorithm}
%\setstretch{1.5}
\LinesNumbered
\SetAlgoNoLine
\DontPrintSemicolon
\caption{$\hat{X} = \texttt{SPF\_BD}(\A,b,s_1,s_2,\mu_1,\mu_2)$}
\label{alg:altminprojbd}
$v_0 \leftarrow \texttt{thres\_init}(\A^* b,\Psi,s_1,s_2,\mu_2)$\;
$t \leftarrow 0$\;
\While{stop condition not satisfied}{
    $t \leftarrow t+1$\;
    $v_{t-1} \leftarrow v_{t-1}/\norm{v_{t-1}}_2$\;
    $\tilde{u}_t \leftarrow \texttt{HTP}(\A_\mathrm{R}(v_{t-1}),b,s_1)$\tcp*{$\A_\mathrm{R}$ defined in (\ref{eq:defAR})}
%    $u_t \leftarrow P_{\Gamma_{s_1} \cap \Phi^{-1} \C_{\mu_1}} \tilde{u}_t$\;
    $u_t \leftarrow \Phi^{-1} P_{\C_{\mu_1} \cap \Phi \Gamma_{s_1}} (\Phi \tilde{u}_t)$\;
    $u_t \leftarrow u_t/\norm{u_t}_2$\;
    $\tilde{v}_t \leftarrow \texttt{HTP}(\A_\mathrm{L}(u_t),b,s_2)$\tcp*{$\A_\mathrm{L}$ defined in (\ref{eq:defAL})}
%    $v_t \leftarrow P_{\Gamma_{s_2} \cap \Psi^{-1} \C_{\mu_2}} \tilde{v}_t$\;
    $v_t \leftarrow \Psi^{-1} P_{\C_{\mu_2} \cap \Psi \Gamma_{s_2}} (\Psi \tilde{v}_t)$\;
}
\Return $\hat{X} = u_t v_t^\transpose$\;
\end{algorithm}

\begin{algorithm}
%\setstretch{1.5}
\LinesNumbered
\SetAlgoNoLine
\DontPrintSemicolon
\caption{$\hat{x} = \texttt{HTP}(A,b,s)$}
\label{alg:htp}
$\hat{x}_0 \leftarrow 0$\;
\While{stop condition not satisfied}{
    $t \leftarrow t+1$\;
    $\widehat{J} \leftarrow \mathrm{supp}\left(P_{\Gamma_s}\left[\hat{x}_{t-1} + \alpha A^*(b - A \hat{x}_{t-1})\right]\right)$\;
    $\hat{x}_t \leftarrow \displaystyle \argmin_x \{ \norm{b - A x}_2^2 :\pl \mathrm{supp}(x) \subset \widehat{J} \}$\;
}
\Return $\hat{x} = x_t$\;
\end{algorithm}

\subsection{Practical implementation with an approximate projection}
\label{subsec:approx_proj}

Computing the exact projection onto the constraint set $\Gamma_{s_1} \cap \Phi^{-1} \C_{\mu_1}$ (resp. $\Gamma_{s_2} \cap \Psi^{-1} \C_{\mu_2}$)
is not a trivial task because of the intersection structure.
To provide a practical algorithm, we propose to replace the exact projection steps of Algorithm~\ref{alg:altminprojbd} (line 7 and 10) by
\[
u_t \leftarrow \texttt{approx\_proj}(\tilde{u}_t,\Phi,s_1,\mu_1)
\]
and
\[
v_t \leftarrow \texttt{approx\_proj}(\tilde{v}_t,\Psi,s_2,\mu_2),
\]
where \texttt{approx\_proj} is described in Algorithm~\ref{alg:approx_proj}.

\begin{algorithm}
%\setstretch{1.5}
\LinesNumbered
\SetAlgoNoLine
\DontPrintSemicolon
\caption{$\hat{u} = \texttt{approx\_proj}(\tilde{u},\Phi,s,\mu)$}
\label{alg:approx_proj}
\eIf{$\tilde{u} \in \Gamma_s$ and $\Phi \tilde{u} \in \C_{\mu}$}{
    \Return $\hat{u} = \tilde{u}$\;
}
{
    $x \leftarrow \Phi \tilde{u}$\;
    \While{stop condition not satisfied}{
        $x \leftarrow P_{\C_\mu} x$\;
        $x \leftarrow P_{\Phi \Gamma_s} x$\;
    }
    \Return $\hat{u} = \Phi^{-1} x$\;
}
\end{algorithm}

Algorithm~\ref{alg:approx_proj} first checks whether the input $\tilde{u}$
belongs to the intersection $\Gamma_s \cap \Phi^{-1} \C_{\mu}$.
When incorporated into Algorithm~\ref{alg:altminprojbd},
the input $\tilde{u}_t$ (resp. $\tilde{v}_t$) already belongs to $\Gamma_{s_1}$ (resp. $\Gamma_{s_2}$).
In particular, as shown in Section~\ref{subsec:priors}, when $\Phi$ is i.i.d. Gaussian,
the spectral flatness level of $\Phi \tilde{u}$ for an arbitrary fixed $\tilde{u} \in \Gamma_{s_1}$ is $O(\log n)$ with high probability.
Therefore, in most iterations of Algorithm~\ref{alg:altminprojbd},
the projection step in Algorithm~\ref{alg:approx_proj} finishes trivially after its first iteration.
We verified this empirically through the simulations in Section~\ref{sec:numres}.

If the initialization of Algorithm~\ref{alg:approx_proj} is outside the intersection,
we apply alternating projections between the two sets $\C_{\mu}$ and $\Phi \Gamma_s$.
First, the projection onto $\Phi \Gamma_s$ is a standard sparse recovery problem and one can use one of the known algorithms.
We use the HTP algorithm (Algorithm~\ref{alg:htp}) for this step in the simulation of Section~\ref{sec:numres}.
Second, for the projection on $\C_\mu$, we derive a computationally efficient algorithm, summarized in Algorithm~\ref{alg:proj_cone}.
As stated in the following theorem, proved in Appendix~\ref{appendix:thm:proj2cone},
Algorithm~\ref{alg:proj_cone} is guaranteed to compute the exact projection onto $\C_{\mu}$.

\begin{theorem}
\label{thm:proj2cone}
The output of Algorithm~\ref{alg:proj_cone} for an input $x \in \cz^n$ is the exact projection of $x$ onto the set $\C_{\mu}$.
\end{theorem}

\begin{algorithm}
%\setstretch{1.5}
\LinesNumbered
\SetAlgoNoLine
\DontPrintSemicolon
\caption{$P_{\C_{\mu}}(x) = \texttt{proj\_c\_mu}(x,\mu)$}
\label{alg:proj_cone}
$\zeta = [\zeta_1,\ldots,\zeta_n]^\transpose \leftarrow F x$\;
$[|\zeta_{(1)}|,\ldots,|\zeta_{(n)}|] = \texttt{sort}([|\zeta_1|,\ldots,|\zeta_n|],\texttt{'descend'})$\;
Find the smallest $k$ such that $(k-1)\mu+\sum_{i=k}^{n}{\frac{|\zeta_{(i)}|^2}{|\zeta_{(k)}|^2}\mu}\geq n$.\;
Construct $\xi \in \cz^n$ such that $\angle \xi_i = \angle \zeta_i$ for $i=1,\ldots,n$ and the $i$th largest magnitude satisfies
\begin{align*}
|\xi_{(i)}| = \begin{cases}
\sqrt{\mu} &\quad\text{if $i\leq k-1$},\\
\sqrt{\frac{n-(k-1)\mu}{\sum_{j=k}^{n} |\zeta_{(j)}|^2}} \pl |\zeta_{(i)}| &\quad\text{if $i \geq k$}.
\end{cases}
\end{align*}
\Return $F^* \norm{\xi}_2^{-2} \xi \xi^* \zeta$\;
\end{algorithm}

\section{Main Results}
\label{sec:main}

We provide performance guarantees for subsampled blind deconvolution at near optimal scaling of the sample complexity
under the scenario that dictionaries $\Phi$ and $\Psi$ are i.i.d. Gaussian.
More precisely, we assume that $\Phi, \Psi \in \cz^{n \times n}$ are mutually independent random matrices,
whose entries are also independent and identically distributed following a zero-mean and complex normal distribution $CN(0,1/n)$.

Our first main result, stated in the next theorem, demonstrates that stable reconstruction is possible at near optimal sample complexity.

To simplify notation, define a set $\calM_{s_1,s_2;\mu_1,\mu_2}$ parameterized by $s_1$, $s_2$, $\mu_1$, and $\mu_2$ as follows:
\begin{equation}
\label{eq:defcalM}
\calM_{s_1,s_2;\mu_1,\mu_2} :=
\{ u v^\transpose \in \cz^{n \times n} :\pl u \in \Gamma_{s_1} \cap \Phi^{-1} \C_{\mu_1}, \pl v \in \Gamma_{s_2} \cap \Psi^{-1} \C_{\mu_2} \}.
\end{equation}
Then, for an element $u v^\transpose$ of $\calM_{s_1,s_2;\mu_1,\mu_2}$,
the sparsity level of $u$ (resp. $v$) is upper-bounded by $s_1$ (resp. $s_2$).
On the other hand, the spectral flatness level of $\Phi u$ (resp. $\Psi v$) is upper-bounded by $\mu_1$ (resp. $\mu_2$).
When $\mu_1 = \infty$, the constraint on the spectral flatness of $\Phi u$ is not active; hence,
\[
\calM_{s_1,s_2;\infty,\mu_2} =
\{ u v^\transpose \in \cz^{n \times n} :\pl u \in \Gamma_{s_1}, \pl v \in \Gamma_{s_2} \cap \Psi^{-1} \C_{\mu_2} \}.
\]
Similarly, we have
\[
\calM_{s_1,s_2;\mu_1,\infty} =
\{ u v^\transpose \in \cz^{n \times n} :\pl u \in \Gamma_{s_1} \cap \Phi^{-1} \C_{\mu_1}, \pl v \in \Gamma_{s_2} \}.
\]

\begin{theorem}
\label{thm:stability}
There exist absolute numerical constants $C > 0$ and $\beta \in \mathbb{N}$ such that the following holds. 
Let $\A: \cz^{n \times n} \to \cz^m$ be defined by (\ref{eq:defcalAbyMell1}) and (\ref{eq:defcalAbyMell2}),
where $\Phi, \Psi \in \cz^{n \times n}$ are independent random matrices whose entries are i.i.d. following $CN(0,1/n)$.
If $m \geq C (\mu_2 s_1 + \mu_1 s_2) \log^5 n$, then with probability $1 - n^{-\beta}$,
all $X \in \calM_{s_1,s_2;\mu_1,\mu_2}$ can be stably reconstructed from $b = \A(X) + z$ in the sense of (\ref{eq:success}).
\end{theorem}

The performance guarantee in Theorem~\ref{thm:stability} applies uniformly to all signals following the underlying model.
In particular, when the parameters $\mu_1$ and $\mu_2$ satisfy $\mu_1 = O(\log n)$ and $\mu_2 = O(\log n)$,
signals $x$ and $y$ of sparsity level $s_1$ and $s_2$ (over corresponding dictionaries) and of spectral flatness level $\mu_1$ and $\mu_2$
are identifiable (up to a scale factor) from $m$ samples of their convolution, where $m$ is proportional (up to a logarithm factor) to $s_1 + s_2$.
This sample complexity has near optimal scaling compared to the number of degrees of freedom in the underlying signal model.

Our second main result asserts that, under an additional condition (peakedness of $u$ and $v$),
the same near optimal performance guarantee is achieved by Algorithm~\ref{alg:altminprojbd}.

Let $\calT_c$ be defined by
\begin{equation}
\label{eq:defcalTc}
\calT_c := \{ u v^\transpose \in \cz^{n \times n} :\pl \norm{u}_\infty \geq c \norm{u}_2, \pl \norm{v}_\infty \geq c \norm{v}_2 \}.
\end{equation}
Then, $\calT_c$ consists of outer products of heavily peaked vectors $u$ and $v$.\footnote{If $\norm{u}_\infty \geq c \norm{u}_2$ for an absolute constant $c \in (0,1)$, there is at least one large element in $u$ which dominates most nonzero elements in magnitude. This motivates the term ``heavy'' peakedness. However, estimating only these dominant components does not provide an acceptable recovery of $u$.}
The heavy peakedness of $u$ and $v$ enables the indices of $u$ and $v$ corresponding to the largest magnitudes to be easily captured by simple thresholding procedures. Under certain conditions, which will be specified later, this will further imply that the estimation error in the initialization by Algorithm~\ref{alg:thres_proj} is bounded below a certain threshold.

\begin{theorem}
\label{thm:altminbd}
%\samepage
There exist absolute numerical constants $c$, $c'$, $C$, $C'$, and $\beta$ such that the following holds. 
Let $\A: \cz^{n \times n} \to \cz^m$ be defined by (\ref{eq:defcalAbyMell1}) and (\ref{eq:defcalAbyMell2}),
where $\Phi, \Psi \in \cz^{n \times n}$ are independent random matrices whose entries are i.i.d. following $CN(0,1/n)$.
Suppose that $\mu_2 \geq C \log n$ and $m \geq C (\mu_1 s_2 + \mu_2 s_1) \log^5 n$.
Then, with probability $1 - n^{-\beta}$, for all $X \in \calM_{s_1,s_2;\mu_1,\mu_2} \cap \calT_c$ and $z \in \cz^m$ satisfying
\begin{equation}
\label{eq:highSNR}
\norm{z}_2 \leq c' \norm{\A(X)}_2,
\end{equation}
Algorithm~\ref{alg:altminprojbd} produces a stable reconstruction of $X$ satisfying (\ref{eq:success}).
\end{theorem}

The conditions $\mu_2 \geq C \log n$ and $\Psi$ is an i.i.d. Gaussian matrix in Theorem~\ref{thm:altminbd}
guarantee that Algorithm~\ref{alg:thres_proj} produces $v_0$ that belongs to $\Psi^{-1} \C_{\mu_2}$.
(The detailed discussion is deferred to Section~\ref{subsec:altminbd_riplike}.)
The initialization in Algorithm~\ref{alg:altminprojbd} is asymmetric in the sense that it involves an estimate of only the right factor $v$.
Therefore, this condition is assumed only for $\mu_2$.

Similarly to Theorem~\ref{thm:stability}, when $\mu_1 = O(\log n)$ and $\mu_2 = O(\log n)$,
the sample complexity in Theorem~\ref{thm:altminbd} reduces to $m = O((s_1+s_2) \log^6 n)$, which is near optimal.

The rest of this section is devoted to proving Theorems~\ref{thm:stability} and \ref{thm:altminbd}.
We first introduce RIP-like properties in the next section.
Then, deterministic performance guarantees given by these RIP-like properties
are presented in Sections~\ref{subsec:stability_riplike} and \ref{subsec:altminbd_riplike}.

\subsection{RIP-like properties}
\label{subsec:riplike}

The previous work \cite{LeeWB2013spf} derived performance guarantees
using the RIP of $\A$ restricted to the set of rank-2 and $(s_1,s_2)$-sparse matrices.
This property holds with high probability at near optimal sample complexity
when each measurement is obtained as an inner product with an i.i.d. subgaussian matrix.
However, the same RIP result is not available for the linear operator $\A$ in blind deconvolution
with random dictionaries (even when the full samples of the convolution are available).
Instead, various RIP-like properties hold with high probability at near optimal sample complexity,
which will lead to performance guarantees for (subsampled) blind deconvolution.
We introduce these RIP-like properties in this section.

The \textit{restricted isometry property} (RIP) was originally proposed
to show the performance guarantee of $\ell_1$-norm minimization in compressed sensing \cite{candes2005decoding}.
We will state our main results in terms of various RIP-like properties,
which are instances of a generalized version of the RIP stated below.

\begin{defi}
Let $(\calH,\hsnorm{\cdot})$ be a Hilbert space where $\hsnorm{\cdot}$ denotes the Hilbert-Schmidt norm.
Let $\calM \subset \calH$ be a centered and symmetric set, i.e.,
$0 \in \calM$ and $\alpha \calM = \calM$ for all $\alpha \in \cz$ of unit modulus.
A linear operator $\A: \calH \to \ell_2^m$ satisfies the $(\calM,\delta)$-RIP if
\[
(1-\delta) \hsnorm{w}^2 \leq \norm{\A(w)}_2^2 \leq (1+\delta) \hsnorm{w}^2, \quad \forall w \in \calM,
\]
or equivalently,
\[
\left| \norm{\A(w)}_2^2 - \hsnorm{w}^2 \right| \leq \delta \hsnorm{w}^2,
\quad \forall w \in \calM.
\]
\end{defi}

Note that $\hsnorm{w}^2 = \langle w, w \rangle$ and $\norm{\A(w)}_2^2 = \langle \A(w), \A(w) \rangle$.
This observation extends RIP to another property called \textit{restricted angle-preserving property} (RAP) defined as follows:

\begin{defi}
Let $\calM,\calM' \subset \calH$ be centered and symmetric sets.
A linear operator $\A: \calH \to \ell_2^m$ satisfies the $(\calM,\calM',\delta)$-RAP if
\[
\left| \langle \A(w'), \A(w)\rangle - \langle w', w\rangle \right| \leq \delta \hsnorm{w} \hsnorm{w'},
\quad \forall w \in \calM, \pl \forall w' \in \calM'.
\]
\end{defi}

In a more restrictive case with orthogonality between $w$ and $w'$ ($\langle w', w \rangle = 0$),
RAP reduces to the \textit{restricted orthogonality property} (ROP) \cite{candes2008restricted}.
\begin{defi}
Let $\calM,\calM' \subset \calH$ be centered and symmetric sets.
A linear operator $\A: \calH \to \ell_2^m$ satisfies the $(\calM,\calM',\delta)$-ROP if
\[
\left| \langle \A(w'), \A(w)\rangle \right| \leq \delta \hsnorm{w} \hsnorm{w'},
\quad \forall w \in \calM, \pl \forall w' \in \calM' ~ \text{s.t.} ~ \langle w', w \rangle = 0.
\]
\end{defi}

\begin{rem}
By definition, $(\calM,\calM,\delta)$-RAP implies $(\calM,\delta)$-RIP. But the converse is not true in general.
\end{rem}

We consider the case where the Hilbert space $\calH$ consists of $n \times n$ matrices.
The main results in the next sections are stated in terms of the following three RIP-like properties
of the linear operator $\A$ defined by (\ref{eq:defcalAbyMell1}) and (\ref{eq:defcalAbyMell2}):
i) $(\calM_{s_1,s_2;\mu_1,\infty},\calM_{s_1,s_2;\mu_1,\infty},\delta)$-RAP; ii) $(\calM_{s_1,s_2;\infty,\mu_2},\calM_{s_1,s_2;\infty,\mu_2},\delta)$-RAP; iii) $(\calM_{s_1,s_2;\mu_1,\infty},\calM_{s_1,s_2;\infty,\mu_2},\delta)$-ROP.
We have proved that these RIP-like properties hold at near optimal sample complexity in a companion paper \cite{LeeJunge2015}, as summarized below.

\begin{theorem}[{\cite[Theorem~1.4]{LeeJunge2015}}]
\label{thm:rap}
There exist absolute numerical constants $C > 0$ and $\beta \in \mathbb{N}$ such that the following holds. 
Let $\A: \cz^{n \times n} \to \cz^m$ be defined by (\ref{eq:defcalAbyMell1}) and (\ref{eq:defcalAbyMell2}),
where $\Phi, \Psi \in \cz^{n \times n}$ are independent random matrices whose entries are i.i.d. following $CN(0,1/n)$.
Suppose $m \geq C \delta^{-2} (\mu_2 s_1 + \mu_1 s_2) \log^5 n$.
Then, with probability at least $1 - n^{-\beta}$,
$\A$ satisfies $(\calM_{s_1,s_2;\mu_1,\infty},\calM_{s_1,s_2;\mu_1,\infty},\delta)$-RAP and $(\calM_{s_1,s_2;\infty,\mu_2},\calM_{s_1,s_2;\infty,\mu_2},\delta)$-RAP.
\end{theorem}

\begin{theorem}[{\cite[Corollary~1.7]{LeeJunge2015}}]
\label{thm:rop}
There exist absolute numerical constants $C > 0$ and $\beta \in \mathbb{N}$ such that the following holds. 
Let $\A: \cz^{n \times n} \to \cz^m$ be defined by (\ref{eq:defcalAbyMell1}) and (\ref{eq:defcalAbyMell2}),
where $\Phi, \Psi \in \cz^{n \times n}$ are independent random matrices whose entries are i.i.d. following $CN(0,1/n)$.
Suppose $m \geq C \delta^{-2} (\mu_2 s_1 + \mu_1 s_2) \log^5 n$.
Then, with probability $1 - n^{-\beta}$, $\A$ satisfies $(\calM_{s_1,s_2;\mu_1,\infty},\calM_{s_1,s_2;\infty,\mu_2},\delta)$-ROP.
\end{theorem}

\subsection{Identifiability of blind deconvolution of sparse signals}
\label{subsec:stability_riplike}

Our first result in the following proposition asserts that
under the aforementioned RIP-like properties, the linear operator $\A$ preserves the distance between two matrices in $\calM_{s_1,s_2;\mu_1,\mu_2}$.

\begin{prop}
\label{prop:stability}
Let $\calM_{s_1,s_2;\mu_1,\mu_2}$ be defined in (\ref{eq:defcalM}).
Suppose that $\A: \cz^{n \times n} \to \cz^m$ satisfies
$(\calM_{s_1,2s_2;\mu_1,\infty},\delta/2)$-RIP, $(\calM_{2s_1,s_2;\infty,\mu_2},\delta/2)$-RIP, and $(\calM_{s_1,2s_2;\mu_1,\infty},\calM_{2s_1,s_2;\infty,\mu_2},\delta/2)$-ROP.
Then,
\begin{align*}
\left(1 - \delta\right) \norm{\widehat{X} - X}_{\mathrm{F}}^2
\leq \norm{\A(\widehat{X} - X)}_2^2
\leq \left(1 + \delta\right) \norm{\widehat{X} - X}_{\mathrm{F}}^2,
\quad \forall X, \widehat{X} \in \calM_{s_1,s_2;\mu_1,\mu_2}.
\end{align*}
In other words, $\A$ satisfies $(\calM_{s_1,s_2;\mu_1,\mu_2}-\calM_{s_1,s_2;\mu_1,\mu_2},\delta)$-RIP.
\end{prop}

Let
\[
\widehat{X} = \argmin_{\widetilde{X} \in \calM_{s_1,s_2;\mu_1,\mu_2}} \norm{b - \A(\widetilde{X})}_2^2.
\]
Then, the $(\calM_{s_1,s_2;\mu_1,\mu_2}-\calM_{s_1,s_2;\mu_1,\mu_2},\delta)$-RIP of $\A$ implies that $\widehat{X}$ satisfies
\begin{equation}
\label{eq:stable}
\frac{\norm{\widehat{X}-X}_{\mathrm{F}}}{\norm{X}_{\mathrm{F}}}
\leq 2\sqrt{\frac{1+\delta}{1-\delta}} \frac{\norm{z}_2}{\norm{\A(X)}_2}.
\end{equation}
Therefore, $\widehat{X}$ is a stable reconstruction of $X$ from the noisy measurements $b = \A(X) + z$.
In the noiseless case ($z = 0$), (\ref{eq:stable}) implies that $X$ is uniquely identified as $\widehat{X}$.

Combining the above result with Theorems~\ref{thm:rap} and \ref{thm:rop} proves Theorem~\ref{thm:stability}.

\subsection{Guaranteed blind deconvolution by alternating minimization}
\label{subsec:altminbd_riplike}

Performance guarantees for blind deconvolution by Algorithm~\ref{alg:altminprojbd}
are given in this section in terms of the RIP-like properties of Section~\ref{subsec:riplike}.
The following conditions are assumed for the analysis in this section:
\begin{description}
  \item[(A1)] Let $b = \A(u v^\transpose) + z$, where $\norm{u}_0 \leq s_1$ and $\norm{v}_0 \leq s_2$.
  \item[(A2)] $\norm{z}_2 \leq \nu \norm{\A(u v^\transpose)}_2$.
  \item[(A3)] $\A$ satisfies $(\calM_{s_1,3s_2;\mu_1,\infty},\calM_{s_1,3s_2;\mu_1,\infty},\delta/2)$-RAP.
  \item[(A4)] $\A$ satisfies $(\calM_{3s_1,s_2;\infty,\mu_2},\calM_{3s_1,s_2;\infty,\mu_2},\delta/2)$-RAP.
  \item[(A5)] $\A$ satisfies $(\calM_{s_1,2s_2;\mu_1,\infty},\calM_{2s_1,s_2;\infty,\mu_2},\delta/2)$-ROP.
\end{description}
The assumption in (A2) implies that the signal-to-noise (SNR) ratio in the measurement vector is higher than a certain threshold.
Although the performance guarantees in Propositions~\ref{prop:conv_w_good_init} and \ref{prop:good_init} require this high SNR condition,
it is empirically observed that the algorithm operates successfully at a moderate SNR (cf. Section~\ref{sec:numres}).

\begin{prop}
\label{prop:conv_w_good_init}
Suppose (A1)-(A5) hold for $\delta = 0.02$ and $\nu = 0.02$.
For an arbitrary $\epsilon > 0$,
the iterates $((u_t,v_t))_{t \in \mathbb{Z}^+}$ by Algorithm~\ref{alg:altminprojbd} satisfy
\begin{equation}
\frac{\fnorm{u_t v_t^\transpose - u v^\transpose}}{\fnorm{u v^\transpose}} \leq 6.28 \frac{\norm{z}_2}{\norm{\A(u v^\transpose)}_2} + \epsilon
\label{eq:conv2set}
\end{equation}
for all $t \geq 2.27 \log (1/\epsilon)$,
provided that the initialization $v_0$ satisfies
\begin{equation}
v_0 \in \Gamma_{s_2} \cap \Psi^{-1} \C_{\mu_2}
\label{eq:feas_initcond}
\end{equation}
and
\begin{equation}
\norm{P_{R(v_0)^\perp} P_{R(v)}} < 0.98.
%\norm{P_{R(v_0)^\perp} P_{R(v)}} < \sin \omega_{\sup}(\delta,\nu),
\label{eq:conv_initcond}
\end{equation}
\end{prop}

\begin{prop}
\label{prop:good_init}
Suppose (A1)-(A5) hold for $\delta = 0.02$ and $\nu = 0.02$.
Suppose that $\norm{u}_\infty \geq \gamma \norm{u}_2$ and $\norm{v}_\infty \geq \gamma \norm{v}_2$ for $\gamma = 0.78$.
Suppose that there exists $s_0 \in \{1,\ldots,s_2\}$ such that (\ref{eq:maxs0}) is satisfied for all $J \subset [n]$ of size $s_0$.
Let $v_0$ denote the output of Algorithm~\ref{alg:thres_proj} for the input $M = \A^*(b)$.
Then, $v_0$ satisfies (\ref{eq:feas_initcond}) and (\ref{eq:conv_initcond}).
\end{prop}

\begin{rem}
The numerical constants for $\delta$, $\nu$, and $\gamma$ in Propositions~\ref{prop:conv_w_good_init} and \ref{prop:good_init} are not optimized.
\end{rem}

\begin{proof}[Proof of Theorem~\ref{thm:altminbd}]
We first show that the condition on $s_0$ in Proposition~\ref{prop:good_init} is satisfied with high probability.
Theorem~\ref{thm:altminbd} assumes that $\mu_2 \geq c_1 \log n$ for an absolute constant $c_1$
and $\Psi \in \cz^{n \times n}$ is a random matrix whose entries are i.i.d. following $CN(0,1/n)$.
Then, there exist absolute constants $c_2, c_3 > 0$ for which the followings hold with overwhelming probability $1-n^{-\beta}$:
i) By an RIP argument (e.g., \cite{baraniuk2008simple}),
we have $\sigma_{\min}(F \Psi_J) > c_2$ for all $J \subset [n]$ of size $|J|$ up to a certain threshold, proportional to $n$;
ii) By a union bound argument and the distribution of a Gaussian random variable,
we also have $\sqrt{n} \norm{F \Psi_J}_{1 \to \infty} \leq \sqrt{n} \norm{F \Psi}_{1 \to \infty} \leq c_3 \log n$ for all $J \subset [n]$.
These two results imply that there exists $s_0$ so that (\ref{eq:maxs0}) is satisfied for all $J \subset [n]$ of size $s_0$.
Note that the other assumptions of Propositions~\ref{prop:conv_w_good_init} and \ref{prop:good_init}
are satisfied directly or via Theorems~\ref{thm:rap} and \ref{thm:rop}.
Therefore, we have proved Theorem~\ref{thm:altminbd}.
\end{proof}

\subsection{Performance guarantee with an approximate projection}

Next, we establish that (under certain technical conditions) the performance guarantee in Proposition~\ref{prop:conv_w_good_init} holds
without requiring the exact computation for the projection steps in Algorithm~\ref{alg:approx_proj}.
Suppose that $u_t$ (resp. $v_t$) is obtained from $\tilde{u}_t$ (resp. $\tilde{v}_t$)
using an approximate projection, e.g., the alternating projections in Algorithm~\ref{alg:approx_proj},
instead of the exact projection in Line 7 (resp. Line 10) of Algorithm~\ref{alg:altminprojbd}.
Suppose that the following conditions are satisfied by the approximate projections in addition to the assumptions of Proposition~\ref{prop:conv_w_good_init}:
\begin{enumerate}
  \item $u_t \in \Gamma_{s_1} \cap \Phi^{-1} \C_{\mu_1}$ and one of the following conditions is satisfied:
  \begin{align*}
  \norm{u_t - \tilde{u}_t}_2 {} & \leq \norm{(v_{t-1}^* v) u - \tilde{u}_t}_2, \\
  \norm{\Phi u_t - \Phi \tilde{u}_t}_2 {} & \leq \norm{\Phi (v_{t-1}^* v) u - \Phi \tilde{u}_t}_2.
  \end{align*}

  \item $v_t \in \Gamma_{s_2} \cap \Psi^{-1} \C_{\mu_2}$ and one of the following conditions is satisfied:
  \begin{align*}
  \norm{v_t - \tilde{v}_t}_2 {} & \leq \norm{(u_t^* u) v - \tilde{v}_t}_2, \\
  \norm{\Psi v_t - \Psi \tilde{v}_t}_2 {} & \leq \norm{\Psi (u_t^* u) v - \Psi \tilde{v}_t}_2.
  \end{align*}
\end{enumerate}
Then, the performance guarantee in Proposition~\ref{prop:conv_w_good_init} remains valid
with the replacement of the exact projections by the approximate projections in the projection steps of Algorithm~\ref{alg:approx_proj}.
Note that the above conditions are satisfied by the exact projections in Algorithm\ref{alg:altminprojbd}.
However, it remains an open question to show whether the they are also satisfied by the alternating projections in Algorithm~\ref{alg:approx_proj}.

\section{Proofs}
\label{sec:proofs}

\subsection{Proof of Proposition~\ref{prop:stability}}

%\begin{proof}[Proof of Proposition~\ref{prop:stability}]
Let $X = u v^\transpose$ and $\widehat{X} = \hat{u} \hat{v}^\transpose$.
Without loss of generality, we may assume that $\hat{u}, v \in \mathbb{S}^{n-1}$.

Let $\rho := \hat{u}^* u$.
Note that
\begin{align}
{} & \norm{\A(\hat{u} \hat{v}^\transpose) - \A(u v^\transpose)}_2^2 \nonumber\\
{} & = \norm{\A(\hat{u} \hat{v}^\transpose) - \rho \A(\hat{u} v^\transpose) + \rho \A(\hat{u} v^\transpose) - \A(u v^\transpose)}_2^2 \nonumber\\
{} & = \norm{\A[\hat{u} (\hat{v}- \rho v)^\transpose] + \A[(\rho \hat{u} - u) v^\transpose]}_2^2 \nonumber\\
{} & = \norm{\A[\hat{u} (\hat{v}- \rho v)^\transpose]}_2^2 + \norm{\A[(\rho \hat{u} - u) v^\transpose]}_2^2 \nonumber\\
{} & \quad + 2 \mathrm{Re}\langle \A[\hat{u} (\hat{v}- \rho v)^\transpose], \A[(\rho \hat{u} - u) v^\transpose] \rangle. \label{eq:proof:prop:stability:lb1}
\end{align}

Since $\hat{u} \in \Gamma_{s_1} \cap \Phi^{-1} \C_{\mu_1}$ and $\hat{v} - \rho v \in \Gamma_{2s_2}$,
we have $\hat{u} (\hat{v}- \rho v)^\transpose \in \calS_{s_1,2s_2;\mu_1,\infty}$.
Therefore, the $(\calM_{s_1,2s_2;\mu_1,\infty},\delta/2)$-RIP of $\A$ implies
\begin{equation}
\label{eq:proof:prop:stability:lb1comp1}
(1-\delta/2) \fnorm{\hat{u} (\hat{v}-\rho v)^\transpose}^2
\leq \norm{\A[\hat{u} (\hat{v}- \rho v)^\transpose]}_2^2
\leq (1+\delta/2) \fnorm{\hat{u} (\hat{v}-\rho v)^\transpose}^2.
\end{equation}

Since $v \in \Gamma_{s_2} \cap \Psi^{-1} \C_{\mu_2}$ and $\rho \hat{u} - u \in \Gamma_{2s_1}$,
by the $(\calM_{2s_1,s_2;\infty,\mu_2},\delta/2)$-RIP of $\A$,
\begin{equation}
\label{eq:proof:prop:stability:lb1comp2}
(1-\delta/2) \fnorm{(\rho \hat{u} - u) v^\transpose}^2
\leq \norm{\A[(\rho \hat{u} - u) v^\transpose]}_2^2
\leq (1+\delta/2) \fnorm{(\rho \hat{u} - u) v^\transpose}^2.
\end{equation}

Since we assumed $\norm{\hat{u}}_2 = 1$, it follows that
\[
\rho \hat{u} = \hat{u} \hat{u}^* u = P_{R(\hat{u})} u.
\]
Therefore,
\[
\langle u - \rho \hat{u}, \hat{u} \rangle = \langle P_{R(\hat{u})^\perp} u, \hat{u} \rangle = 0.
\]

Similar to the previous two cases, the $(\calM_{s_1,2s_2;\mu_1,\infty},\calM_{2s_1,s_2;\infty,\mu_2},\delta/2)$-ROP of $\A$ implies
that the magnitude of the last term in (\ref{eq:proof:prop:stability:lb1}) is upper-bounded by
\begin{equation}
\label{eq:proof:prop:stability:lb1comp3}
|\langle \A[\hat{u} (\hat{v}- \rho v)^\transpose], \A[(\rho \hat{u} - u) v^\transpose] \rangle|
\leq (\delta/2) \norm{\hat{u} (\hat{v} - \rho v)^\transpose}_{\mathrm{F}} \norm{(\rho \hat{u} - u) v^\transpose}_{\mathrm{F}}.
\end{equation}

By plugging (\ref{eq:proof:prop:stability:lb1comp1}), (\ref{eq:proof:prop:stability:lb1comp2}), and (\ref{eq:proof:prop:stability:lb1comp3})
into (\ref{eq:proof:prop:stability:lb1}), we continue as follows:
\begin{align*}
{} & \norm{\A(\hat{u} \hat{v}^\transpose) - \A(u v^\transpose)}_2^2 \\
{} & \geq (1-\delta/2) \norm{\hat{u} (\hat{v} - \rho v)^\transpose}_{\mathrm{F}}^2 + (1-\delta/2) \norm{(\rho \hat{u} - u) v^\transpose}_{\mathrm{F}}^2 \\
{} & \quad - \delta \norm{\hat{u} (\hat{v} - \rho v)^\transpose}_{\mathrm{F}} \norm{(\rho \hat{u} - u) v^\transpose}_{\mathrm{F}} \\
{} & \geq (1-\delta/2) \norm{\hat{u} (\hat{v} - \rho v)^\transpose}_{\mathrm{F}}^2 + (1-\delta/2) \norm{(\rho \hat{u} - u) v^\transpose}_{\mathrm{F}}^2 \\
{} & \quad - (\delta/2) \left(\norm{\hat{u} (\hat{v} - \rho v)^\transpose}_{\mathrm{F}}^2 + \norm{(\rho \hat{u} - u) v^\transpose}_{\mathrm{F}}^2\right) \\
{} & = (1-\delta) \left(\norm{\hat{u} (\hat{v} - \rho v)^\transpose}_{\mathrm{F}}^2 + \norm{(\rho \hat{u} - u) v^\transpose}_{\mathrm{F}}^2\right) \\
{} & = (1-\delta) \norm{\hat{u} (\hat{v} - \rho v)^\transpose + (\rho \hat{u} - u) v^\transpose}_{\mathrm{F}}^2 \\
{} & = (1-\delta) \norm{\hat{u} \hat{v}^\transpose - u v^\transpose}_{\mathrm{F}}^2,
\end{align*}
where the second step holds by the inequality of arithmetic and geometric means,
and fourth step follows since $\langle (\rho \hat{u} - u) v^\transpose, \hat{u} (\hat{v} - \rho v)^\transpose \rangle = 0$.

By symmetry, we also have
\begin{align*}
{} & \norm{\A(\hat{u} \hat{v}^\transpose) - \A(u v^\transpose)}_2^2 \\
{} & \leq (1+\delta/2) \norm{\hat{u} (\hat{v} - \rho v)^\transpose}_{\mathrm{F}}^2 + (1+\delta/2) \norm{(\rho \hat{u} - u) v^\transpose}_{\mathrm{F}}^2 \\
{} & \quad + \delta \norm{\hat{u} (\hat{v} - \rho v)^\transpose}_{\mathrm{F}} \norm{(\rho \hat{u} - u) v^\transpose}_{\mathrm{F}} \\
{} & \leq (1+\delta) \norm{\hat{u} \hat{v}^\transpose - u v^\transpose}_{\mathrm{F}}^2.
\end{align*}
%\end{proof}

\subsection{Proof of Proposition~\ref{prop:conv_w_good_init}}

Due to the similarity between Algorithm~\ref{alg:altminprojbd} and the SPF algorithm,
the performance guarantee for iterative updates in Algorithm~\ref{alg:altminprojbd} is derived
by modifying the corresponding part in the analysis of SPF \cite[Theorem~III.6 and its proof in Section VII.A]{LeeWB2013spf}.

We first show the convergence of the angle between $u_t$ and $u$ (resp. $v_t$ and $v$) given by
\[
\sin[\angle(u_t, u)] = \norm{P_{R(u_t)^\perp} P_{R(u)}}
\quad
\text{and}
\quad
\sin[\angle(v_t, v)] = \norm{P_{R(v_t)^\perp} P_{R(v)}}.
\]

In the previous work \cite{LeeWB2013spf}, it was first shown that
the RIP of $\A$ for all rank-2 and $(3s_1,3s_2)$-sparse matrices implies
the $(\Gamma_{3s_2},\delta)$-RIP of $\A_\mathrm{L}(\hat{u})$ for all $\hat{u} \in \Gamma_{s_1}$
and the $(\Gamma_{3s_1},\delta)$-RIP of $\A_\mathrm{R}(\hat{v})$ for all $\hat{v} \in \Gamma_{s_2}$.
Then, by an RIP-based performance guarantee of HTP, the decay of the angle is implied these RIPs \cite[Lemma~VII.3]{LeeWB2013spf}.

However, in subsampled blind deconvolution, the linear operator $\A$ does not provide the aforementioned RIP.
Instead, it provides the RIP-like properties in (A3)--(A5).
These RIP-like properties imply the RIPs of $\A_\mathrm{L}(\hat{u})$ and $\A_\mathrm{R}(\hat{v})$ for all $\hat{u} \in \Gamma_{s_1}$ and $\hat{v} \in \Gamma_{s_2}$.
More precisely, the followings hold: i) $(\calM_{s_1,3s_2,\mu_1,\infty},\calM_{s_1,3s_2,\mu_1,\infty},\delta)$-RAP of $\A$ implies $(\Gamma_{3s_2},\delta)$-RIP of $\A_\mathrm{L}(\hat{u})$ for all $\hat{u} \in \mathbb{S}^{n-1} \cap \Gamma_{s_1} \cap \Psi^{-1} \C_{\mu_1}$; ii) $(\calM_{3s_1,s_2,\infty,\mu_2},\calM_{3s_1,s_2,\infty,\mu_2},\delta)$-RAP of $\A$ implies $(\Gamma_{3s_1},\delta)$-RIP of $\A_\mathrm{R}(\hat{v})$ for all $\hat{v} \in \mathbb{S}^{n-1} \cap \Gamma_{s_2} \cap \Psi^{-1} \C_{\mu_2}$.
Together with the ROP $(\calM_{s_1,2s_2,\mu_1,\infty},\calM_{2s_1,s_2,\infty,\mu_2},\delta)$-ROP of $\A$,
we obtain Lemma~\ref{lemma:update_uv}, which provides recursive formulas for the sequences of angles.

\begin{lemma}[{Analog of \cite[Lemma~VII.3]{LeeWB2013spf} for blind deconvolution}]
\label{lemma:update_uv}
Assume the setup of Proposition~\ref{prop:conv_w_good_init}.
The iterates $((u_t,v_t))_{t \in \mathbb{Z}^+}$ by Algorithm~\ref{alg:altminprojbd} satisfy
\begin{equation}
\label{eq:update_u}
\sin[\angle(u_t, u)] \leq \frac{2 C_\delta^{\text{\rm HTP}} \sqrt{1+\delta}}{\sqrt{1-\delta}} \left( \delta \tan[\angle(v_{t-1}, v)] + (1+\delta) \sec[\angle(v_{t-1}, v)] \frac{\norm{z}_2}{\norm{\A(u v^\transpose)}_2} \right)
\end{equation}
and
\begin{equation}
\label{eq:update_v}
\sin[\angle(v_t, v)] \leq \frac{2 C_\delta^{\text{\rm HTP}} \sqrt{1+\delta}}{\sqrt{1-\delta}} \left( \delta \tan[\angle(u_t, u)] + (1+\delta) \sec[\angle(u_t, u)] \frac{\norm{z}_2}{\norm{\A(u v^\transpose)}_2} \right),
\end{equation}
where $C_\delta^{\text{\rm HTP}}$ is defined by
\begin{equation}
\label{eq:CdeltaHTP}
C_\delta^{\text{\rm HTP}} :=
\frac{\sqrt{2(1-\delta)}+\sqrt{1+\delta}}{\sqrt{1-\delta}(\sqrt{1-\delta^2}-\sqrt{2\delta^2})}.
\end{equation}
\end{lemma}

In fact, the choice of $\delta$ in Proposition~\ref{prop:conv_w_good_init} implies $2 \delta C_\delta^{\text{\rm HTP}} \sqrt{1+\delta} / \sqrt{1-\delta} < 1$.

\begin{proof}[Proof of Lemma~\ref{lemma:update_uv}]

We only prove the part in (\ref{eq:update_u}) that bounds the estimation error in $u_t$.
The proof for the other part in (\ref{eq:update_v}) follows directly by symmetry.

Recall that the previous iterate $v_{t-1}$ is set to satisfy $v_{t-1} \in \Psi^{-1} \C_{\mu_2} \cap \Gamma_{s_2}$.
(This is assumed for $v_0$ and satisfied by the projection step for $t \geq 1$.)
Furthermore, $v_{t-1}$ is normalized to have the unit $\ell_2$ norm, i.e., $v_{t-1} \in \mathbb{S}^{n-1}$.

The measurement vector $b$ is rewritten as
\begin{align*}
b {} & = \A(u v^\transpose) + z = [\A_\mathrm{R}(v)] u + z \\
{} & = [\A_\mathrm{R}(P_{R(v_{t-1})} v + P_{R(v_{t-1})^\perp} v)] u + z \\
{} & = [\A_\mathrm{R}(P_{R(v_{t-1})} v)] u  + [\A_\mathrm{R}(P_{R(v_{t-1})^\perp} v)] u + z \\
{} & = [\A_\mathrm{R}(v_{t-1})] \{(v_{t-1}^* v) u\} + [\A_\mathrm{R}(P_{R(v_{t-1})^\perp} v)] u + z \pl.
\end{align*}

Since $v_{t-1} \in \mathbb{S}^{n-1} \cap \Gamma_{s_2} \cap \Psi^{-1} \C_{\mu_2}$,
it follows from the $(\calM_{3s_1,s_2,\infty,\mu_2},\calM_{3s_1,s_2,\infty,\mu_2},\delta)$-RAP of $\A$
that $\A_\mathrm{R}(v_{t-1})$ satisfies the $(\Gamma_{3s_1},\delta)$-RIP.

Note that $P_{R(v_{t-1})^\perp} v \in \Gamma_{2s_2}$ follows from $v,v_{t-1} \in \Gamma_{s_2}$.
Let $J_1$ be an arbitrary subset of $[n]$ with $|J_1| \leq s_1$.
Then, we have
\begin{equation}
\label{eq:proof:lemma:update_uv:ub_angle}
\begin{aligned}
{} & \norm{\Pi_{J_1} [\A_\mathrm{R}(v_{t-1})]^* [\A_\mathrm{R}(P_{R(v_{t-1})^\perp} v)] u}_2 \\
{} & = \max_{\breve{u} \in \mathbb{S}^{n-1} \cap \Gamma_{s_1}}
\left| \langle \breve{u},[\A_\mathrm{R}(v_{t-1})]^* [\A_\mathrm{R}(P_{R(v_{t-1})^\perp} v)] u \rangle \right| \\
{} & = \max_{\breve{u} \in \mathbb{S}^{n-1} \cap \Gamma_{s_1}}
\left| \langle \A(\breve{u} v_{t-1}^*), \A(u v^* P_{R(v_{t-1})^\perp}) \rangle \right| \\
{} & \leq \delta \fnorm{\breve{u} v_{t-1}^*} \fnorm{u v^* P_{R(v_{t-1})^\perp}} \\
{} & = \delta \norm{u}_2 \norm{P_{R(v_{t-1})^\perp} v}_2,
\end{aligned}
\end{equation}
where the inequality holds by $(\calM_{s_1,2s_2,\mu_1,\infty},\calM_{s_1,s_2,\infty,\mu_2},\delta)$-ROP of $\A$
since $\breve{u} v_{t-1}^* \in \calM_{s_1,s_2,\infty,\mu_2}$, $u v^* P_{R(v_{t-1})^\perp} \in \calM_{s_1,2s_2,\mu_1,\infty}$,
and $\langle \breve{u} v_{t-1}^*, u v^* P_{R(v_{t-1})^\perp} \rangle = 0$.

An RIP-based performance guarantee of HTP \cite[Theorem~3.8]{Fou2011htp}\footnote{In fact, it is also possible to show that HTP converges in finite steps to $\tilde{u}_t$ with a comparable error bound, where only $C_\delta^{\text{\rm HTP}}$ increases by a constant factor (c.f. \cite[Lemma~7.1]{LeeWB2013spf}).} implies that the iterates in HTP converges to $\tilde{u}_t$ satisfying
\[
\norm{\tilde{u}_t - (v_{t-1}^* v) u}_2
\leq C_\delta^{\text{\rm HTP}} \norm{\Pi_{\widehat{J}} [\A_\mathrm{R}(v_{t-1})]^* ([\A_\mathrm{R}(P_{R(v_{t-1})^\perp} v)] u + z)}_2,
\]
where $\widehat{J}$ denotes the support of $\tilde{u}_t$ and the constant $C_\delta^{\text{\rm HTP}}$ is given by
\begin{equation}
\label{eq:Cdelta}
C_\delta^{\text{\rm HTP}} =
\frac{\sqrt{2(1-\delta)}+\sqrt{1+\delta}}{\sqrt{1-\delta}(\sqrt{1-\delta^2}-\sqrt{2\delta^2})}.
\end{equation}

Therefore, (\ref{eq:proof:lemma:update_uv:ub_angle}) and the $(\Gamma_{3s_1},\delta)$-RIP of $\A_\mathrm{R}(v_{t-1})$ provide
\begin{equation}
\label{eq:proof:lemma:update_uv:bndmse}
\norm{\tilde{u}_t - (v_{t-1}^* v) u}_2
\leq C_\delta^{\text{\rm HTP}} (\delta \norm{u}_2 \norm{P_{R(v_{t-1})^\perp} v}_2 + \sqrt{1+\delta} \norm{z}_2).
\end{equation}

Recall that $u_t = \Phi^{-1} P_{\C_{\mu_1} \cap \Phi \Gamma_{s_1}} (\Phi \tilde{u}_t)$.
By the optimality of the projection and the fact that $u \in \Phi^{-1} \C_{\mu_1} \cap \Gamma_{s_1}$,
it follows that
\[
\norm{\Phi u_t - \Phi \tilde{u}_t}_2 \leq \norm{\Phi (v_{t-1}^* v) u - \Phi \tilde{u}_t}_2.
\]

We show that $\Phi$ satisfies $(\Gamma_{2s_1},\delta)$-RIP under the assumption of Proposition~\ref{prop:conv_w_good_init}.
In fact, since
\[
\Phi u = e_1 \conv \Phi u,
\]
$\Phi$ corresponds to $\A_\mathrm{R}(e_1)$ where $\Psi = I_n$.
Since $\norm{F e_1}_\infty = n^{-1/2}$, $e_1 \in \C_{\mu_2}$ for any $\mu_2 \geq 1$.
Therefore, by \cite[Corollary~1.5]{LeeJunge2015}, $\Phi$ satisfies $(\Gamma_{2s_1},\delta)$-RIP.

By the above results, we have
\begin{equation}
\label{eq:proof:lemma:update_uv:ubmse}
\begin{aligned}
{} & \norm{\Phi u_t - \Phi (v_{t-1}^* v) u}_2 \\
{} & \leq \norm{\Phi u_t - \Phi \tilde{u}_t}_2 + \norm{\Phi \tilde{u}_t - \Phi (v_{t-1}^* v) u}_2 \\
{} & \leq 2 \norm{\Phi \tilde{u}_t - \Phi (v_{t-1}^* v) u}_2 \\
{} & = 2 \norm{\Phi \{\tilde{u}_t - (v_{t-1}^* v) u\}}_2 \\
{} & \leq 2 \sqrt{1+\delta} \norm{\tilde{u}_t - (v_{t-1}^* v) u}_2,
\end{aligned}
\end{equation}
where the last step follows from the fact that $\tilde{u}_t - (v_{t-1}^* v) u \in \Gamma_{2s_1}$ and $\Phi$ satisfies $(\Gamma_{2s_1},\delta)$-RIP.

On the other hand, since $u_t - (v_{t-1}^* v) u \in \Gamma_{2s_1}$, we have
\begin{equation}
\label{eq:proof:lemma:update_uv:lbmse}
\norm{\Phi u_t - \Phi (v_{t-1}^* v) u}_2 \geq \sqrt{1-\delta} \norm{u_t - (v_{t-1}^* v) u}_2.
\end{equation}

Combining (\ref{eq:proof:lemma:update_uv:bndmse}), (\ref{eq:proof:lemma:update_uv:ubmse}), and (\ref{eq:proof:lemma:update_uv:lbmse}) gives
\begin{equation}
\label{eq:proof:lemma:update_uv:bnd2}
\begin{aligned}
{} & \norm{u_t - (v_{t-1}^* v) u}_2 \\
{} & \leq \frac{2 C_\delta^{\text{\rm HTP}} \sqrt{1+\delta}}{\sqrt{1-\delta}}
(\delta \norm{u}_2 \norm{P_{R(v_{t-1})^\perp} v}_2 + \sqrt{1+\delta} \norm{z}_2) \\
{} & \leq \frac{2 C_\delta^{\text{\rm HTP}} \sqrt{1+\delta}}{\sqrt{1-\delta}}
(\delta \norm{u}_2 \sin[\angle(v_{t-1},v)] \norm{v}_2 + \sqrt{1+\delta} \norm{z}_2).
\end{aligned}
\end{equation}

The left-hand-side of (\ref{eq:proof:lemma:update_uv:bnd2}) is bounded from below as
\begin{equation}
\label{eq:proof:lemma:update_uv:bnd3}
\norm{u_t - (v_{t-1}^* v) u}_2
\geq \norm{P_{R(u_t)^\perp} (v_{t-1}^* v) u}_2
= \sin[\angle(u_t,u)] \norm{u}_2 \cos[\angle(v_{t-1},v)] \norm{v}_2.
\end{equation}

Combining (\ref{eq:proof:lemma:update_uv:bnd2}) and (\ref{eq:proof:lemma:update_uv:bnd3}) gives
\[
\sin[\angle(u_t,u)]
\leq \frac{2 C_\delta^{\text{\rm HTP}} \sqrt{1+\delta}}{\sqrt{1-\delta}}
\left(\delta \tan[\angle(v_{t-1},v)] + \sqrt{1+\delta} \frac{\norm{z}_2}{\norm{u}_2 \norm{v}_2} \sec[\angle(v_{t-1},v)]\right).
\]

Finally, by Theorem~\ref{prop:stability}, we have
\[
\frac{\norm{z}_2}{\norm{u}_2 \norm{v}_2}
= \frac{\norm{z}_2}{\norm{u v^\transpose}_\mathrm{F}}
\leq \sqrt{1+\delta} \frac{\norm{z}_2}{\norm{\A(u v^\transpose)}_2},
\]
which finishes the proof.
\end{proof}

Lemma~\ref{lemma:update_uv} implies the convergence of the angle via the following lemma.

\begin{lemma}[{\cite[Lemma~7.7]{LeeWB2013spf}}]
\label{lemma:convtheta}
Assume the setup of Proposition~\ref{prop:conv_w_good_init}.
Define
\begin{equation}
\label{eq:defOmega}
\Omega := \left\{ \omega \in [0, \pi/2) : \sin \omega \geq
\frac{2 C_\delta^{\text{\rm HTP}} \sqrt{1+\delta}}{\sqrt{1-\delta}} \left[ \delta \tan \omega + (1+\delta) \nu \sec \omega \right] \right\}.
\end{equation}
(Note that $\Omega$ is not empty when $\delta$ and $\nu$ are set as in Proposition~\ref{prop:conv_w_good_init}.)
If $\angle (v_0,v) < \sup \Omega$, then
\[
\lim_{t \to \infty} \max(0, \angle(u_t, u) - \inf \Omega) = 0
\]
and the convergence rate is linear.
Moreover,
\[
\sin (\inf \Omega) \leq C_1 \frac{\norm{z}_2}{\norm{\A(u v^\transpose)}_2}.
\]
where $C_1$ given by
\[
C_1 = \frac{\sqrt{1-\delta}(1+\delta)}{\sqrt{1-\delta} \cos (\inf \Omega) - 2 \delta C_\delta^{\text{\rm HTP}} \sqrt{1+\delta}}.
\]
\end{lemma}

Then, it remains to bound the estimation error $\fnorm{u v^\transpose - u_t v_t^\transpose}$.
Without loss of generality, we may assume $\norm{u}_2 = 1$.
When $v_t$ is updated, $u_t$ is normalized in the $\ell_2$ norm.
The error is decomposed as
\[
u v^\transpose - u_t v_t^\transpose
= P_{R(u_t)} u v^\transpose - u_t v_t^\transpose + P_{R(u_t)^\perp} u v^\transpose
= u_t (\langle u_t, u \rangle v - v_t)^\transpose + P_{R(u_t)^\perp} u v^\transpose.
\]

As shown in the proof of Lemma~\ref{lemma:update_uv}, we have
\[
\norm{v_t - \langle u_t, u \rangle v}_2
\leq \frac{2 C_\delta^{\text{\rm HTP}} \sqrt{1+\delta}}{\sqrt{1-\delta}} (\delta \norm{P_{R(u_t)^\perp} u}_2 \norm{v}_2 + \sqrt{1+\delta} \norm{z}_2).
\]

Therefore, via the Pythagorean theorem, we derive an upper bound of the estimation error given by
\begin{align*}
\fnorm{u v^\transpose - u_t v_t^\transpose}^2
{} & = \fnorm{u_t (\langle u_t, u \rangle v - v_t)^\transpose}^2 + \fnorm{P_{R(u_t)^\perp} u v^\transpose}^2 \\
{} & = \norm{\langle u_t, u \rangle v - v_t}_2^2 + \norm{P_{R(u_t)^\perp} u}_2 \norm{v}_2^2 \\
{} & \leq \left[1 + \left(\frac{2 \delta C_\delta^{\text{\rm HTP}} \sqrt{1+\delta}}{\sqrt{1-\delta}}\right)^2\right] \norm{P_{R(u_t)^\perp} P_{R(u)}}^2 \fnorm{u v^\transpose}^2 + 4 (1+\delta) (C_\delta^{\text{\rm HTP}})^2 \norm{z}_2^2,
\end{align*}
which implies
\[
\fnorm{u v^\transpose - u_t v_t^\transpose}
\leq \sqrt{1 + \left(\frac{2 \delta C_\delta^{\text{\rm HTP}} \sqrt{1+\delta}}{\sqrt{1-\delta}}\right)^2} \norm{P_{R(u_t)^\perp} P_{R(u)}} \fnorm{u v^\transpose} + 2 \sqrt{1+\delta} C_\delta^{\text{\rm HTP}} \norm{z}_2.
\]

Finally, by applying Lemma~\ref{lemma:convtheta}, we get
\begin{align*}
{} & \limsup_{t \to \infty} 
\frac{\fnorm{u v^\transpose - u_t v_t^\transpose}}{\fnorm{u v^\transpose}} \\
{} & \leq \limsup_{t \to \infty} \sqrt{1 + \left(\frac{2 \delta C_\delta^{\text{\rm HTP}} \sqrt{1+\delta}}{\sqrt{1-\delta}}\right)^2} \norm{P_{R(u_t)^\perp} P_{R(u)}}
+ 2 (1+\delta) C_\delta^{\text{\rm HTP}} \frac{\norm{z}_2}{\norm{\A(u v^\transpose)}_2} \\
{} & \leq C_1 \left[\sqrt{1 + \left(\frac{2 \delta C_\delta^{\text{\rm HTP}} \sqrt{1+\delta}}{\sqrt{1-\delta}}\right)^2} + 2 (1+\delta) C_\delta^{\text{\rm HTP}}\right] \frac{\norm{z}_2}{\norm{\A(u v^\transpose)}_2}.
\end{align*}

\subsection{Proof of Proposition~\ref{prop:good_init}}

First, we show that $v_0$ satisfies the feasibility condition in (\ref{eq:feas_initcond}).
Note that if (\ref{eq:maxs0}) is satisfied for a particular set $\widehat{J}$,
then it is also satisfied for all subsets of $\widehat{J}$.
Since (\ref{eq:maxs0}) is satisfied for all $J \subset [n]$ of size $s_0$,
indeed, it is satisfied for all $J \subset [n]$ of size up to $s_0$.
In particular, (\ref{eq:maxs0}) is satisfied for $J = \widehat{J}_2$,
where $\widehat{J}_2$ is the particular set explored in Algorithm~\ref{alg:thres_proj}.
Recall that Algorithm~\ref{alg:thres_proj} choose $\widehat{J}_2$ so that $|\widehat{J}_2| \leq s_2$.
Therefore, as was shown in Section~\ref{sec:algorithm},
the estimate $v_0$ by Algorithm~\ref{alg:thres_proj} belongs to the constraint set $\Gamma_{s_2} \cap \Psi^{-1} \C_{\mu_2}$.

Next, we show that $v_0$ satisfies (\ref{eq:conv_initcond}).
This part of the proof is similar to the corresponding part of the previous work \cite[Section~7.2]{LeeWB2013spf}.
In the derivation of the performance guarantees for SPF (\cite[Section~7.2]{LeeWB2013spf} and \cite[Appendix~B]{LeeWB2013spf}),
there were steps implied by the RIP of $\A$ that holds all rank-two and $(2s_1,2s_2)$-sparse matrices at near optimal sample complexity.
However, the linear operator $\A$ in subsampled blind deconvolution does not provide this property.
In this proof, we show that these steps hold by the RIP-like properties and (\ref{eq:maxs0}) as a series of lemmas.
Then, we will obtain the desired property in (\ref{eq:conv_initcond}).

\begin{lemma}
\label{lemma:riplike_cross}
Assume the setup of Proposition~\ref{prop:good_init}. Then,
\begin{equation}
\label{eq:proof:lemma:spf_lemma7.10:ineq1}
\norm{\Pi_{\widehat{J}_1} [(\A^*\A - \id)(u v^\transpose)] \Pi_{\widehat{J}_2}}_{2 \to 2} \leq \delta \norm{u v^\transpose}_{\mathrm{F}}.
\end{equation}
\end{lemma}

\begin{proof}[Proof of Lemma~\ref{lemma:riplike_cross}]
First, we verify that (\ref{eq:proof:lemma:spf_lemma7.10:ineq1}) is derived from
the $(\calM_{s_1,s_2,\infty,\mu_2},\calM_{s_1,s_2,\infty,\mu_2},\delta)$-RAP of $\A$ as follows:
\begin{align*}
{} & \norm{\Pi_{\widehat{J}_1} [(\A^*\A - \id)(u v^\transpose)] \Pi_{\widehat{J}_2}}_{2 \to 2} \\
{} & \leq \max_{\breve{u} \in \Gamma_{s_1} \cap \mathbb{S}^{n-1}}
\max_{\begin{subarray}{c} \breve{v} \in \mathbb{S}^{n-1} \\ \supp{\breve{v}} \subset \widehat{J}_2 \end{subarray}}
|\langle \breve{u} \breve{v}^\transpose ,\pl (\A^*\A - \id)(u v^\transpose) \rangle| \\
{} & \leq \max_{\breve{u} \in \Gamma_{s_1} \cap \mathbb{S}^{n-1}}
\max_{\begin{subarray}{c} \breve{v} \in \mathbb{S}^{n-1} \\ \supp{\breve{v}} \subset \widehat{J}_2 \end{subarray}}
\delta \norm{u v^\transpose}_\mathrm{F} \norm{\breve{u} \breve{v}^\transpose}_\mathrm{F} \\
{} & = \delta \norm{u v^\transpose}_\mathrm{F},
\end{align*}
where the first step follows from the definition of the spectral norm,
and the second step from the fact that all $s_0$-sparse $\breve{v}$ supported on $\widehat{J}_2$ belong to $\Psi^{-1} \C_{\mu_2}$ since $s_0$ satisfies (\ref{eq:maxs0}).
\end{proof}

\begin{lemma}
\label{lemma:riplike_noise}
Assume the setup of Proposition~\ref{prop:good_init}. Then,
\begin{equation}
\label{eq:proof:lemma:spf_lemma7.10:ineq2}
\norm{\Pi_{\widehat{J}_1} [\A^*(z)] \Pi_{\widehat{J}_2}}_{2 \to 2} \leq \sqrt{1+\delta} \norm{z}_2.
\end{equation}
\end{lemma}

\begin{proof}[Proof of Lemma~\ref{lemma:riplike_noise}]
Similarly, (\ref{eq:proof:lemma:spf_lemma7.10:ineq2}) is derived as follows:
\begin{align*}
{} & \norm{\Pi_{\widehat{J}_1} [\A^*(z)] \Pi_{\widehat{J}_2}}_{2 \to 2} \\
{} & = \max_{\breve{u} \in \Gamma_{s_1} \cap \mathbb{S}^{n-1}}
\max_{\begin{subarray}{c} \breve{v} \in \mathbb{S}^{n-1} \\ \supp{\breve{v}} \subset \widehat{J}_2 \end{subarray}}
|\langle \breve{u} \breve{v}^\transpose ,\pl \A^*(z) \rangle| \\
{} & = \max_{\breve{u} \in \Gamma_{s_1} \cap \mathbb{S}^{n-1}}
\max_{\begin{subarray}{c} \breve{v} \in \mathbb{S}^{n-1} \\ \supp{\breve{v}} \subset \widehat{J}_2 \end{subarray}}
|\langle \A(\breve{u} \breve{v}^\transpose) ,\pl z \rangle| \\
{} & = \max_{\breve{u} \in \Gamma_{s_1} \cap \mathbb{S}^{n-1}}
\max_{\begin{subarray}{c} \breve{v} \in \mathbb{S}^{n-1} \\ \supp{\breve{v}} \subset \widehat{J}_2 \end{subarray}}
|\breve{u}^* [\A_{\mathrm{R}}(\breve{v})]^* z| \\
{} & \leq \max_{\breve{u} \in \Gamma_{s_1} \cap \mathbb{S}^{n-1}}
\max_{\begin{subarray}{c} \breve{v} \in \mathbb{S}^{n-1} \\ \supp{\breve{v}} \subset \widehat{J}_2 \end{subarray}}
\sqrt{1+\delta} \norm{\breve{u}}_2 \norm{\breve{v}}_2 \norm{z}_2 \\
{} & = \sqrt{1+\delta} \norm{z}_2,
\end{align*}
where the inequality holds since
the $(\calM_{s_1,s_2,\infty,\mu_2},\calM_{s_1,s_2,\infty,\mu_2},\delta)$-RAP of $\A$
implies the $(\Gamma_{s_1},\delta)$-RIP of $\A_{\mathrm{R}}(\breve{v})$ for all $\breve{v} \in \Gamma_{s_2} \cap \Psi^{-1} \C_{\mu_2} \cap \mathbb{S}^{n-1}$,
which follows from the fact that all $\breve{v}$ supported on $\widehat{J}_2$ belong to $\Gamma_{s_2} \cap \Psi^{-1} \C_{\mu_2}$
since $s_0$ satisfies (\ref{eq:maxs0}).
\end{proof}

Next, we derive (\ref{eq:conv_initcond}) using Lemmas~\ref{lemma:riplike_cross} and \ref{lemma:riplike_noise}.
The first step is to compute an upper bound on the angle between $v_0$ and $v$ using
the non-Hermitian $\sin\theta$ theorem \cite[pp. 102--103]{Wed1972perturbation}.
This step is analogous to \cite[Lemma~7.10]{LeeWB2013spf} and summarize as the following lemma.

\begin{lemma}[{Analog of \cite[Lemma~7.10]{LeeWB2013spf} for blind deconvolution}]
Under the setup of Proposition~\ref{prop:good_init}, we have
\label{lemma:7.10}
\begin{equation}
\label{eq:init_step4}
\norm{P_{R(v_0)^\perp} P_{R(v)}}
\leq \frac{(\norm{\Pi_{\widehat{J}_1} u}_2 / \norm{u}_2) (\norm{\Pi_{\widehat{J}_2}^\perp v}_2 / \norm{v}_2) + \delta + (1+\delta) \nu}{(\norm{\Pi_{\widehat{J}_1} u}_2 / \norm{u}_2) - \delta - (1+\delta) \nu}.
\end{equation}
\end{lemma}

\begin{proof}[Proof of Lemma~\ref{lemma:7.10}]

Let $Z := \Pi_{\widehat{J}_1} u v^\transpose$.
Let $\widehat{Z}$ denote the best rank-1 approximation of $\Pi_{\widehat{J}_1} [\A^*(b)] \Pi_{\widehat{J}_2}$ in the spectral norm.
Let $E := \Pi_{\widehat{J}_1} [\A^*(b)] \Pi_{\widehat{J}_2} - Z$.
Then,
\begin{align*}
\norm{Z + E - \widehat{Z}}
{} & = \norm{\Pi_{\widehat{J}_1} [\A^*(b)] \Pi_{\widehat{J}_2} - \widehat{Z} } \\
{} & \leq \norm{ \Pi_{\widehat{J}_1} [\A^*(b)] \Pi_{\widehat{J}_2} - \Pi_{\widehat{J}_1} u v^\transpose \Pi_{\widehat{J}_2}} \\
{} & = \norm{\Pi_{\widehat{J}_1} [(\A^* \A - \id)(u v^\transpose)] \Pi_{\widehat{J}_2} + \Pi_{\widehat{J}_1} [\A^*(z)] \Pi_{\widehat{J}_2}} \\
{} & \leq \norm{\Pi_{\widehat{J}_1} [(\A^* \A - \id)(u v^\transpose)] \Pi_{\widehat{J}_2}} + \norm{\Pi_{\widehat{J}_1} [\A^*(z)] \Pi_{\widehat{J}_2}} \\
{} & \leq \delta \fnorm{u v^\transpose} + \sqrt{1+\delta} \norm{z}_2 \\
{} & \leq \fnorm{u v^\transpose} [\delta  + (1+\delta) \nu],
\end{align*}
where the second inequality follows from Lemmas~\ref{lemma:riplike_cross} and \ref{lemma:riplike_noise}, and the last step holds since
\begin{equation}
\label{eq:ub_noise}
\norm{z}_2 \leq \nu \norm{\A(u v^\transpose)}_2 \leq \nu \sqrt{1+\delta} \fnorm{u v^\transpose}.
\end{equation}

Similarly, the difference $E$ is bounded in the spectral norm by
\begin{align*}
\norm{E}
{} & = \norm{\Pi_{\widehat{J}_1} [\A^*(b)] \Pi_{\widehat{J}_2} - \Pi_{\widehat{J}_1} u v^\transpose (\Pi_{\widehat{J}_2} + \Pi_{\widehat{J}_2}^\perp)} \\
{} & = \norm{\Pi_{\widehat{J}_1} [(\A^* \A - \id)(u v^\transpose)] \Pi_{\widehat{J}_2} - \Pi_{\widehat{J}_1} u v^\transpose \Pi_{\widehat{J}_2}^\perp + \Pi_{\widehat{J}_1} [\A^*(z)] \Pi_{\widehat{J}_2}} \\
{} & \leq \norm{\Pi_{\widehat{J}_1} [(\A^* \A - \id)(u v^\transpose)] \Pi_{\widehat{J}_2}} + \norm{\Pi_{\widehat{J}_1} u v^\transpose \Pi_{\widehat{J}_2}^\perp} + \norm{\Pi_{\widehat{J}_1} [\A^*(z)] \Pi_{\widehat{J}_2}} \\
{} & \leq \delta \fnorm{u v^\transpose} + \norm{\Pi_{\widehat{J}_1} u v^\transpose \Pi_{\widehat{J}_2}^\perp} + \sqrt{1+\delta} \norm{z}_2 \\
{} & \leq \fnorm{u v^\transpose} \left[
\delta + (\norm{\Pi_{\widehat{J}_1} u}_2 / \norm{u}_2) (\norm{\Pi_{\widehat{J}_2}^\perp v}_2 / \norm{v}_2) + (1+\delta) \nu
\right],
\end{align*}
where the second inequality follows from Lemmas~\ref{lemma:riplike_cross} and \ref{lemma:riplike_noise},
and the last step follows from \ref{eq:ub_noise}.

Finally, note that $\norm{Z}$ is rewritten as $\norm{Z} = \norm{\Pi_{\widehat{J}_1} u v^\transpose} = \norm{\Pi_{\widehat{J}_1} u}_2 \norm{v}_2$.

Since $R(Z^\transpose) = R(v)$ and $R(v_0) = R(\widehat{Z}^\transpose)$, the proof completes by applying Lemma~\ref{lemma:sintheta}.

\begin{lemma}[{Non-Hermitian $\sin\theta$ theorem \cite{Wed1972perturbation}: rank-1 case}]
\label{lemma:sintheta}
Let $Z \in \cz^{n \times d}$ be a rank-1 matrix.
Let $\widehat{Z} \in \cz^{n \times d}$ be the best rank-1 approximation of $Z + E$ in the Frobenius norm.
Then,
\[
\norm{P_{R(\widehat{Z}^\transpose)^\perp} P_{R(Z^\transpose)}}
\leq \frac{\max(\norm{P_{R(Z)} E},\norm{E P_{R(Z^*)}})}{\norm{Z} - \norm{Z + E - \widehat{Z}}}.
\]
\end{lemma}

\end{proof}

By Lemma~\ref{lemma:7.10}, we get a sufficient condition for (\ref{eq:conv_initcond}) given by
\begin{equation}
\label{eq:suff_conv_initcond}
\delta + (1+\delta) \nu \leq
\frac{(\norm{\Pi_{\widehat{J}_1} u}_2 / \norm{u}_2) \left(\sin (\sup \Omega)) - (\norm{\Pi_{\widehat{J}_2}^\perp v}_2 / \norm{v}_2) \right)}{1 + \sin (\sup \Omega)},
\end{equation}
where the set $\Omega$ is defined in (\ref{eq:defOmega}).
Recall that $\sin (\sup \Omega)$ depends only on $\delta$ and $\nu$.

Similar to \cite[Section~3.2]{LeeWB2013spf}, there exist $\delta$, $\nu$, $\gamma$ that satisfy
\begin{equation}
\label{eq:suff_conv_initcond_simple}
\delta + \nu + \delta \nu \leq
\frac{\gamma^2 + \sin (\sup \Omega) - \sqrt{2}}{3 + \sin (\sup \Omega)}.
\end{equation}
For example, the choice of $\delta$, $\nu$, and $\gamma$ in Proposition~\ref{prop:good_init} satisfies (\ref{eq:suff_conv_initcond_simple}).

It remains to show (\ref{eq:suff_conv_initcond_simple}) implies (\ref{eq:suff_conv_initcond}).
We first show that $\widehat{J}_1$ includes the index of the largest element of $u$ in magnitude using the following lemma.

\begin{lemma}[{Analog of \cite[Lemma~7.12]{LeeWB2013spf} for blind deconvolution}]
\label{lemma:7.12}
Under the setup of Proposition~\ref{prop:good_init}, suppose that there exists $\widetilde{J}_1 \subset \supp{u}$ satisfying
\begin{equation}
\label{eq:lemma:7.12}
2 \delta + 2(1 + \delta) \nu < \min_{j \in \widetilde{J}_1} |\langle e_j, u \rangle|.
\end{equation}
Then, $\widetilde{J}_1 \subset \widehat{J}_1$.
\end{lemma}

\begin{proof}[Proof of Lemma~\ref{lemma:7.12}]
The proof is identical to that of \cite[Lemma~7.12]{LeeWB2013spf} except that
Lemmas~\ref{lemma:riplike_cross} and \ref{lemma:riplike_noise} replace the RIP of $\A$;
hence, we omit the proof and refer the details to \cite[Appendix~B.3]{LeeWB2013spf}.
\end{proof}

Let $j_0$ denote the index for the largest entry of $u$ in magnitude. Then, $\norm{\Pi_{\{j_0\}} u}_2 = \norm{u}_\infty$.
We apply Lemma~\ref{lemma:7.12} with $\widetilde{J}_1 = \{j_0\}$,
Since (\ref{eq:suff_conv_initcond_simple}) implies (\ref{eq:lemma:7.12}),
it follows that $j_0$ is included in $\widehat{J}_1$.

Then, we continue deriving a sufficient condition for (\ref{eq:suff_conv_initcond}).
Let $k_0$ denote the index for the largest entry of $v$ in magnitude.
Recall that $\widehat{J}_2$ is selected so that
\[
\widehat{J}_2 = \argmax_{\widetilde{J} \subset [n_2], |\widetilde{J}| \leq s_0} \fnorm{\Pi_{\widehat{J}_1} [\A^*(b)] \Pi_{\widetilde{J}}}.
\]
Therefore,
\begin{equation}
\label{eq:res_ell2_thres}
\fnorm{\Pi_{\widehat{J}_1} [\A^*(b)] \Pi_{\widehat{J}_2}}
\geq \fnorm{\Pi_{\widehat{J}_1} [\A^*(b)] \Pi_{\{k_0\}}}
\geq \fnorm{\Pi_{\{j_0\}} [\A^*(b)] \Pi_{\{k_0\}}},
\end{equation}
where the last step holds since $j_0 \in \widehat{J}_1$.

By Lemmas~\ref{lemma:riplike_cross} and \ref{lemma:riplike_noise},
the left-hand-side of (\ref{eq:res_ell2_thres}) is upper-bounded by
\[
\fnorm{\Pi_{\widehat{J}_1} u v^\transpose \Pi_{\widehat{J}_2}} + \delta \fnorm{u v^\transpose} + \sqrt{1+\delta} \norm{z}_2,
\]
and the right-hand-side of (\ref{eq:res_ell2_thres}) is lower-bounded by
\[
\fnorm{\Pi_{\{j_0\}} u v^\transpose \Pi_{\{k_0\}}} - \delta \fnorm{u v^\transpose} - \sqrt{1+\delta} \norm{z}_2.
\]
Therefore,
\[
\norm{\Pi_{\widehat{J}_1} u}_2 \norm{\Pi_{\widehat{J}_2} v}_2
\geq \norm{u}_\infty \norm{v}_\infty - 2 \delta \fnorm{u v^\transpose} - \sqrt{1+\delta} \norm{z}_2,
\]
which implies
\begin{equation}
\label{eq:init_step3}
(\norm{\Pi_{\widehat{J}_1} u}_2 \norm{u}_2) (\norm{\Pi_{\widehat{J}_2} v}_2 / \norm{v}_2)
\geq (\norm{u}_\infty / \norm{u}_2) (\norm{v}_\infty / \norm{v}_2) - 2 \delta - (1+\delta) \nu.
\end{equation}

Finally, we verify that (\ref{eq:suff_conv_initcond_simple}) together with (\ref{eq:init_step3}) implies (\ref{eq:suff_conv_initcond}).
This procedure is identical to the corresponding part in the previous work and we refer the details to \cite[p. 30, Section~7.2.2]{LeeWB2013spf}.

\section{Numerical Results}
\label{sec:numres}
In this section, we conduct numerical experiments to test the empirical performance of the alternating minimization algorithm (Algorithm~\ref{alg:altminprojbd}). In particular, we verify the following aspects:
\begin{enumerate}
  \item When the measurement is noiseless, the largest sparsity level, for which the reconstruction in Algorithm~\ref{alg:altminprojbd} is successful, is proportional to the length $m$ of the measurement vector.
  \item When the measurement is noisy, Algorithm~\ref{alg:altminprojbd} can produce stable reconstruction in the same optimal regime.
\end{enumerate}

In all the experiments, we synthesize the dictionaries $\Phi,\Psi\in\mathbb{R}^{n\times n}$ such that the entries are independent and identically distributed following a Gaussian distribution $N(0,1/n)$. The coefficient vectors $u,v\in\mathbb{R}^n$ have $s=s_1=s_2$ nonzero entries. The supports are chosen independently and uniformly at random, and the nonzero entries are independent and identically distributed following a Gaussian distribution $N(0,1)$. The measurement vector is synthesized using \eqref{eq:defcalA}, and the reconstruction $\hat{X}$ using Algorithm~\ref{alg:altminprojbd} is observed. The signal-to-noise ratio of the measurement is computed as
\[
\mathrm{SNR} := -20 \log_{10}\left(\frac{\|z\|_2}{\|\mathcal{A}(X)\|_2}\right).
\]
The reconstruction is declared successful if the reconstruction signal-to-distortion ratio (RSDR), defined by
\[
\mathrm{RSDR} := -20 \log_{10}\left(\frac{\|\hat{X}-X\|_F}{\|X\|_F}\right),
\]
is larger than certain threshold. We conduct experiments under three noise levels, where the measurement SNR's are $\infty$ (noiseless), $40$dB, and $20$dB, respectively. The cut-off thresholds of RSDR for the three noise levels are $60$dB, $30$dB, and $10$dB, respectively.

We first consider the measurement model without subsampling, where every element of the convolution is observed and $m=n$. To observe how the largest sparsity level $s$ for successful reconstruction varies with $m$, we compute the empirical success rate for different values of $m$ and $s$. The results for three noise levels are shown in Figures \ref{fig:successrate_a}, \ref{fig:successrate_b}, and \ref{fig:successrate_c}.
We also conduct an experiment for the measurement model with subsampling. The sampling operator $S_\Omega$ is uniform subsampling by a factor of $2$. In this case, \emph{every other} element of the convolution is observed and the number of measurements is $m=n/2$. The result for noiseless measurements is shown in Figure \ref{fig:successrate_d}.

Judging by the results in Figures \ref{fig:successrate_a} and \ref{fig:successrate_d}, we can tell that, with or without subsampling in the measurements, when the length $m$ of the measurement vector is large, the largest sparsity level $s$ for successful reconstruction is proportional to $m$. This observation verifies the prediction of our theoretical analysis.
Comparing the results from Figures \ref{fig:successrate_a}, \ref{fig:successrate_b}, and \ref{fig:successrate_c}, we find that Algorithm~\ref{alg:altminprojbd} is stable even with medium to high levels of noise. Although the reconstruction SDR decreases with the measurement SNR, the regime of successful reconstruction largely remains unchanged. Therefore, despite the high SNR requirement in the theoretical guarantee (Propositions~\ref{prop:conv_w_good_init} and \ref{prop:good_init}, and Theorem~\ref{thm:altminbd}), the algorithm performs consistently as the noise level increases.

\begin{figure}[htbp]
  \centering
  \subfloat[\label{fig:successrate_a}]{
\begin{tikzpicture}[scale=0.6]

\begin{axis}[%
width=4in,
height=2in,
scale only axis,
point meta min=0,
point meta max=1,
axis on top,
xmin=0.5,
xmax=16.5,
xtick={2,4,6,8,10,12,14,16},
xticklabels={{2048},{4096},{6144},{8192},{10240},{12288},{14336},{16384}},
ymin=0.5,
ymax=7.5,
ytick={1,3,5,7},
yticklabels={{1/256},{3/256},{5/256},{7/256}},
axis background/.style={fill=white}
]
\addplot [forget plot] graphics [xmin=0.5,xmax=16.5,ymin=0.5,ymax=7.5] {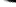};
\end{axis}
\end{tikzpicture}%
  }\quad
  \subfloat[\label{fig:successrate_b}]{
\begin{tikzpicture}[scale=0.6]

\begin{axis}[%
width=4in,
height=2in,
scale only axis,
point meta min=0,
point meta max=1,
axis on top,
xmin=0.5,
xmax=16.5,
xtick={2,4,6,8,10,12,14,16},
xticklabels={{2048},{4096},{6144},{8192},{10240},{12288},{14336},{16384}},
ymin=0.5,
ymax=7.5,
ytick={1,3,5,7},
yticklabels={{1/256},{3/256},{5/256},{7/256}},
axis background/.style={fill=white}
]
\addplot [forget plot] graphics [xmin=0.5,xmax=16.5,ymin=0.5,ymax=7.5] {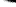};
\end{axis}
\end{tikzpicture}%
  }\\
  \subfloat[\label{fig:successrate_c}]{
\begin{tikzpicture}[scale=0.6]

\begin{axis}[%
width=4in,
height=2in,
scale only axis,
point meta min=0,
point meta max=1,
axis on top,
xmin=0.5,
xmax=16.5,
xtick={2,4,6,8,10,12,14,16},
xticklabels={{2048},{4096},{6144},{8192},{10240},{12288},{14336},{16384}},
ymin=0.5,
ymax=7.5,
ytick={1,3,5,7},
yticklabels={{1/256},{3/256},{5/256},{7/256}},
axis background/.style={fill=white}
]
\addplot [forget plot] graphics [xmin=0.5,xmax=16.5,ymin=0.5,ymax=7.5] {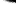};
\end{axis}
\end{tikzpicture}%
  }\quad
  \subfloat[\label{fig:successrate_d}]{
\begin{tikzpicture}[scale=0.6]

\begin{axis}[%
width=4in,
height=2in,
scale only axis,
point meta min=0,
point meta max=1,
axis on top,
xmin=0.5,
xmax=16.5,
xtick={2,4,6,8,10,12,14,16},
xticklabels={{1024},{2048},{3072},{4096},{5120},{6144},{7168},{8192}},
ymin=0.5,
ymax=7.5,
ytick={1,3,5,7},
yticklabels={{1/256},{3/256},{5/256},{7/256}},
axis background/.style={fill=white}
]
\addplot [forget plot] graphics [xmin=0.5,xmax=16.5,ymin=0.5,ymax=7.5] {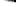};
\end{axis}
\end{tikzpicture}%
  }
  \caption{Empirical success rate. The $x$-axis represents the length $m$, and the $y$-axis represents the ratio $s/m$. The empirical success rate is represented by the grayscale (white for one and black for zero). (a) The measurement is noiseless. (b) The measurement SNR is $40$dB. (c) The measurement SNR is $20$dB. (d) The measurement is noiseless and uniformly subsampled by a factor of $2$.}
\label{fig:successrate}
\end{figure}

\section{Conclusion}
\label{sec:concl}

Near optimal performance guarantees for the subsampled blind deconvolution problem are studied in this paper.
Since the pioneering work that achieved a near optimal performance guarantee for blind deconvolution under subspace priors \cite{ahmed2014blind},
attempts have been made to involve a more general union of subspaces model, which correspond to a sparsity model with unknown support.
Even without subsampling, which is indeed a difficult component that makes the recovery in blind deconvolution significantly more challenging,
existing theoretic analyses were not able to provide a performance guarantee at near optimal relation among model parameters.
Similarly to the previous near optimal performance guarantee for blind deconvolution given a subspace model \cite{ahmed2014blind},
we assumed a stronger prior than the conventional sparsity, which also enforces spectral flatness on the unknown signals.
Under this prior, we proposed an iterative algorithm for recovery in subsampled blind deconvolution
and derived its performance guarantee, which shows stable reconstruction with high probability at near optimal sample complexity.

\section*{Acknowledgement}
K. Lee thanks A. Ahmed and F. Krahmer for discussions, which inspired the random dictionary model in this paper. 
This work was supported in part by the National Science Foundation under Grants CCF 10-18789, DMS 12-01886, and IIS 14-47879.

\appendix

\section{Proof of Proposition~\ref{Proposition 2.1}}
\label{appendix:prop2.1}
%\begin{proof}[Proof of Proposition~\ref{Proposition 2.1}]
Note that because $F$ is a unitary matrix, $F \Phi$ has the same distribution as $\Phi$.
By an RIP result on a matrix with gaussian distributed entries (cf. \cite{baraniuk2008simple}),
it follows that if (\ref{A1}) is satisfied, then with high probability, we have
\begin{equation}
\label{eq:bnd_sf_den}
\sup_{u \in B_2^n \cap \Gamma_s} |\norm{F \Phi u}_2^2 - 1| \leq \eta,
\end{equation}
where $B_2^n$ is the Euclidean unit ball in $\cz^n$.
On the other hand, since $B_2^n \cap \Gamma_s \subset \sqrt{s} B_1^n$,
where $B_1^n$ denotes the $\ell_1$ unit ball in $\cz^n$, we have
\begin{align*}
\sup_{u \in B_2^n \cap \Gamma_s} n \norm{F \Phi u}_\infty^2
\leq \sup_{u \in \sqrt{s} B_1^n} n \norm{F \Phi u}_\infty^2
= s \max_{i,j} |g_{i,j}|^2,
\end{align*}
where $g_{i,j}$ denotes the $(i,j)$th entry of $\sqrt{n} F \Phi$.
Note that $|g_{i,j}|^2$'s are i.i.d. following a Chi-squared distribution with one degree of freedom,
whose tail distribution $P(|g_{i,j}|^2 > z)$ is upper-bounded by $\sqrt{ze^{1-z}}$.
The tail distribution of the order statistic $\max_{i,j} |g_{i,j}|^2$ is then upper-bounded by
\[
1 - (1 - \sqrt{ze^{1-z}})^{n^2} \leq n^2 \sqrt{ze^{1-z}}.
\]
Therefore, with probability $1 - n^{-\beta}$,
\begin{equation}
\label{eq:bnd_sf_num}
\sup_{u \in B_2^n \cap \Gamma_s} n \norm{F \Phi u}_\infty^2 \leq s \max_{i,j} |g_{i,j}|^2 \leq c_1 s \log n
\end{equation}
for absolute constants $c_1, \beta > 0$.

Combining (\ref{eq:bnd_sf_den}) and (\ref{eq:bnd_sf_num}) completes the proof.
%\end{proof}

\section{Proof of Proposition~\ref{Proposition 2.2}}
\label{appendix:prop2.2}
%\begin{proof}[Proof of Proposition~\ref{Proposition 2.2}]
As shown in the proof of Proposition~\ref{Proposition 2.1},
if (\ref{A1}) is satisfied, then (\ref{eq:bnd_sf_den}) holds with high probability, which implies
\[
\norm{F \Phi u}_2^2 \leq 1+\eta, \quad \forall u \in B_2^n \cap \Gamma_s.
\]

On the other hand, by fixing the support of $u$ to $\{j_1,\ldots,j_s\}$ and choosing the $i$th row of $F \Phi$
in the maximization with respect to $u$, we get a lower bound for $\sup_{u \in B_2^n \cap \Gamma_s} n \norm{F \Phi u}_\infty^2$ given by
\[
\sup_{u \in B_2^n \cap \Gamma_s} n \norm{F \Phi u}_\infty^2
\geq \sum_{k=1}^s |g_{i,j_k}|^2,
\]
where $g_{i,j_k}$ is a standard Gaussian random variable obtained as the $j_k$th entry of the $i$th row of $\sqrt{n} F \Phi$.
Note that this lower bound is a Chi-squared random variable with $s$ degrees of freedom;
hence, its mean and median are $s$ and $s[1-2/(9s)]^3$, respectively.
In other words, with probability 0.5, there exists $u \in B_2^n \cap \Gamma_s$ such that $n \norm{F \Phi u}_\infty^2 \geq s[1-2/(9s)]^3$.
%\end{proof}

\section{Proof of Proposition~\ref{Proposition 2.3}}
\label{appendix:prop2.3}
%\begin{proof}[Proof of Proposition~\ref{Proposition 2.3}]
Without loss of generality, we may assume $\norm{u}_2 = 1$.
Then, $n \norm{F \Phi u}_2^2$ and the squared magnitude of each entry of $\sqrt{n} F \Phi u$
correspond to Chi-square random variables with $n$ and $1$ degrees of freedom, respectively.

The cumulative density function of $n \norm{F \Phi u}_2^2$ is upper-bounded by
\[
P(n \norm{F \Phi u}_2^2 < (1-\alpha) n) \leq [(1-\alpha) e^\alpha]^{n/2}
\]
for $0 < \alpha < 1$. Note that $0 < (1-\alpha) e^\alpha < 1$.
Therefore, except with exponentially decaying probability in $n$,
\begin{equation}
\label{eq:prop2.3:num}
\frac{1}{\norm{F \Phi u}_2^2} > \frac{1}{1-\alpha}.
\end{equation}

Similar to the proof of Proposition~\ref{Proposition 2.1}, via a tail distribution of an order statistic,
we have
\[
P(n \norm{F \Phi u}_\infty^2 > z) \leq 1 - (1 - \sqrt{ze^{1-z}})^n \leq n \sqrt{ze^{1-z}}.
\]
Therefore, with probability $1 - n^{-\beta}$,
\begin{equation}
\label{eq:prop2.3:den}
n \norm{F \Phi u}_\infty^2 \leq c_1 s \log n
\end{equation}
for absolute constants $c_1, \beta > 0$.
Combining (\ref{eq:prop2.3:num} and (\ref{eq:prop2.3:den}) completes the proof.
%\end{proof}

\section{Proof of Theorem~\ref{thm:proj2cone}}
\label{appendix:thm:proj2cone}
%\begin{proof}[Proof of Theorem~\ref{thm:proj2cone}]
Note that the image of $\C_{\mu}$ by the DFT is another cone given by
\begin{equation}
\label{eq:defcheckCmu}
\widecheck{\C}_{\mu} := F \C_{\mu} = \{w \in \cz^n :\pl \norm{w}_\infty \leq \sqrt{\mu/n} \norm{w}_2 \}.
\end{equation}
Since $F$ is a unitary matrix, which preserves $\ell_2$ norm, the projection satisfies:
\[
P_{\C_\mu}(x) = F^* P_{\widecheck{\C}_\mu}(F x).
\]
Therefore, it suffices to show that $\norm{\xi}_2^{-2} \xi \xi^* \zeta$ is the projection of $\zeta = F x$ onto the cone $\widecheck{\C}_\mu$.

Since $\widecheck{\C}_\mu$ is a cone, the condition $P_{\widecheck{\C}_\mu}(\zeta) = \norm{\xi}_2^{-2} \xi \xi^* \zeta$ is equivalent to
\begin{equation}
\label{eq:proj_cone_eq1}
\xi \in \arg\max_{\tilde{\xi} \in \widecheck{\C}_\mu} \frac{|\langle \tilde{\xi}, \zeta \rangle|}{\norm{\tilde{\xi}}_2}.
\end{equation}

When the magnitudes of variable $\tilde{\xi}$ is fixed,
the denominator of the objective function in (\ref{eq:proj_cone_eq1}) is fixed regardless of the phase of $\tilde{\xi}$, 
and the numerator is maximized by choosing the phase of $\tilde{\xi}$ identical to that of $\zeta$.
In other words, a maximizer to (\ref{eq:proj_cone_eq1}) and $\zeta$ have the same phase.

Algorithm~\ref{alg:proj_cone} constructs $\xi$ so that $\xi$ and $\zeta$ have the same phase.
Therefore, the aforementioned necessary condition for the optimality of $\xi$ is satisfied.
It remains to show the optimality of the magnitudes of $\xi$.

Let $a = [a_1,\ldots,a_n]^\transpose$ where $a_i = |\xi_i|^2$ for $i=1,\ldots,n$.
Since $\xi$ and $\zeta$ have the same phase, we only need to show that
\begin{equation}
\label{eq:proj_cone_eq2}
a \in \arg\max_{\tilde{a} \in \widecheck{\C}_\mu} \frac{\sum_{i=1}^n |\zeta_i| \sqrt{\tilde{a}_i}}{\norm{\tilde{a}}_2},
\end{equation}
where $\tilde{a} = [\tilde{a}_1,\ldots,\tilde{a}_n]^\transpose$.

In Algorithm~\ref{alg:proj_cone}, only the projection onto the span of $\xi$ is used.
Therefore, without loss of generality, we may assume that $\norm{\xi}_2^2 = \sum_{i=1}^n a_i = n$.
Then, (\ref{eq:proj_cone_eq2}) is satisfied
if $a$ is a minimizer to the following convex optimization problem
\begin{equation}
\label{eq:proj_cone_eq3}
\begin{array}{ll}
\displaystyle \min_{\tilde{a} \in \mathbb{R}_+^n} & - \sum_{i=1}^n {|\zeta_i| \sqrt{\tilde{a}_i}} \\
\text{s.t.} & \sum_{i=1}^n {\tilde{a}_i} = n, \\
& 0 \leq \tilde{a}_i \leq \mu, \quad \forall i=1,\ldots,n.
\end{array}
\end{equation}

The vector $\xi$ in algorithm~\ref{alg:proj_cone} satisfies $\norm{\xi}_2^2 = \sum_{i=1}^n {a_i} = n$.
Furthermore, the magnitudes of $\xi$ satisfy
\[
|\xi_i| =
\begin{cases}
\sqrt{\mu}, & \text{if $|\zeta_i| > 2\sqrt{\mu}\lambda$}, \\
\frac{|\zeta_i|}{2\lambda}, & \text{if $|\zeta_i|\leq 2\sqrt{\mu}\lambda$},
\end{cases}
\qquad \forall i=1,\ldots,n,
\]
for some $\lambda > 0$.
This implies that $\tilde{a}$ satisfies the KKT condition (for some $[\mu_1,\ldots,\mu_n]^\transpose \in \mathbb{R}_+^n$):
\begin{align*}
& \sqrt{a_i} = \frac{|\zeta_i|}{2(\lambda+\mu_i)},\\
& \text{if $\mu_i >0$, then $\sqrt{a_i} = \sqrt{\mu}$}= \frac{|\zeta_i|}{2(\lambda+\mu_i)}< \frac{|\zeta_i|}{2\lambda},\\
& \text{if $\mu_i =0$, then $\sqrt{a_i} = \frac{|\zeta_i|}{2\lambda}$},
\end{align*}
for $i=1,\ldots,n$.
Therefore, $a$ is a minimizer to (\ref{eq:proj_cone_eq3}), which completes the proof.
%\end{proof}

\bibliographystyle{IEEEtran}
\bibliography{IEEEabrv,lra,linalg,cs,bdconv,preprint,appl}

\begin{thebibliography}{10}
\providecommand{\url}[1]{#1}
\csname url@samestyle\endcsname
\providecommand{\newblock}{\relax}
\providecommand{\bibinfo}[2]{#2}
\providecommand{\BIBentrySTDinterwordspacing}{\spaceskip=0pt\relax}
\providecommand{\BIBentryALTinterwordstretchfactor}{4}
\providecommand{\BIBentryALTinterwordspacing}{\spaceskip=\fontdimen2\font plus
\BIBentryALTinterwordstretchfactor\fontdimen3\font minus
  \fontdimen4\font\relax}
\providecommand{\BIBforeignlanguage}[2]{{%
\expandafter\ifx\csname l@#1\endcsname\relax
\typeout{** WARNING: IEEEtran.bst: No hyphenation pattern has been}%
\typeout{** loaded for the language `#1'. Using the pattern for}%
\typeout{** the default language instead.}%
\else
\language=\csname l@#1\endcsname
\fi
#2}}
\providecommand{\BIBdecl}{\relax}
\BIBdecl

\bibitem{kundar1996blind}
D.~Kundur and D.~Hatzinakos, ``Blind image deconvolution,'' \emph{{IEEE} Signal
  Process. Mag.}, vol.~13, no.~3, pp. 43--64, May 1996.

\bibitem{liu2005relevant}
Y.~Lin and D.~Lee, ``Relevant deconvolution for acoustic source estimation,''
  in \emph{Proc. ICASSP}, Mar. 2005, pp. v/529--v/532.

\bibitem{kazemi2014sparse}
N.~Kazemi and M.~D. Sacchi, ``Sparse multichannel blind deconvolution,''
  \emph{Geophysics}, vol.~79, no.~5, pp. V143--V152, 2014.

\bibitem{kotera2013blind}
J.~Kotera, F.~{\v{S}}roubek, and P.~Milanfar, ``Blind deconvolution using
  alternating maximum a posteriori estimation with heavy-tailed priors,'' in
  \emph{Comput. Anal. Images Patterns}.\hskip 1em plus 0.5em minus 0.4em\relax
  Springer, 2013, pp. 59--66.

\bibitem{sroubek2007unified}
F.~{\v{S}}roubek, G.~Crist{\'o}bal, and J.~Flusser, ``A unified approach to
  superresolution and multichannel blind deconvolution,'' \emph{{IEEE} Trans.
  Image Process.}, vol.~16, no.~9, pp. 2322--2332, 2007.

\bibitem{ying2007joint}
L.~Ying and J.~Sheng, ``Joint image reconstruction and sensitivity estimation
  in {SENSE} ({JSENSE}),'' \emph{Magnet. Reson. Med.}, vol.~57, no.~6, pp.
  1196--1202, 2007.

\bibitem{abed1997blind}
K.~Abed-Meraim, W.~Qiu, and Y.~Hua, ``Blind system identification,''
  \emph{Proc. {IEEE}}, vol.~85, no.~8, pp. 1310--1322, 1997.

\bibitem{ahmed2014blind}
A.~Ahmed, B.~Recht, and J.~Romberg, ``Blind deconvolution using convex
  programming,'' \emph{{IEEE} Trans. Inf. Theory}, vol.~60, no.~3, pp.
  1711--1732, Mar. 2014.

\bibitem{ling2015self}
S.~Ling and T.~Strohmer, ``Self-calibration and biconvex compressive sensing,''
  \emph{arXiv preprint arXiv:1501.06864}, 2015.

\bibitem{chi2015guaranteed}
Y.~Chi, ``Guaranteed blind sparse spikes deconvolution via lifting and convex
  optimization,'' \emph{arXiv preprint arXiv:1506.02751}, 2015.

\bibitem{LeeWB2013spf}
K.~Lee, Y.~Wu, and Y.~Bresler, ``Near optimal compressed sensing of sparse
  rank-one matrices via sparse power factorization,'' \emph{arXiv preprint
  arXiv:1312.0525}, 2013.

\bibitem{choudhary2014identifiability}
S.~Choudhary and U.~Mitra, ``Identifiability scaling laws in bilinear inverse
  problems,'' \emph{arXiv preprint arXiv:1402.2637}, 2014.

\bibitem{candes2010power}
E.~J. Cand{\`e}s and T.~Tao, ``The power of convex relaxation: Near-optimal
  matrix completion,'' \emph{{IEEE} Trans. Inf. Theory}, vol.~56, no.~5, pp.
  2053--2080, 2010.

\bibitem{choudhary2014sparse}
S.~Choudhary and U.~Mitra, ``Sparse blind deconvolution: What cannot be done,''
  in \emph{Proc. ISIT}, Jun. 2014, pp. 3002--3006.

\bibitem{LeeJunge2015}
K.~Lee and M.~Junge, ``{RIP}-like properties in subsampled blind
  deconvolution,'' \emph{arXiv preprint arXiv:1511.06146}.

\bibitem{li2015unified}
Y.~Li, K.~Lee, and Y.~Bresler, ``A unified framework for identifiability
  analysis in bilinear inverse problems with applications to subspace and
  sparsity models,'' \emph{arXiv preprint arXiv:1501.06120}, 2015.

\bibitem{li2015identifiability}
------, ``Identifiability in blind deconvolution with subspace or sparsity
  constraints,'' \emph{arXiv preprint arXiv:1505.03399}, 2015.

\bibitem{Fou2011htp}
S.~Foucart, ``Hard thresholding pursuit: an algorithm for compressive
  sensing,'' \emph{SIAM J. Numer. Anal.}, vol.~49, no.~6, pp. 2543--2563, 2011.

\bibitem{needell2009cosamp}
D.~Needell and J.~A. Tropp, ``{CoSaMP}: Iterative signal recovery from
  incomplete and inaccurate samples,'' \emph{Appl. Comput. Harmon. Anal.},
  vol.~26, no.~3, pp. 301--321, 2009.

\bibitem{dai2009subspace}
W.~Dai and O.~Milenkovic, ``Subspace pursuit for compressive sensing signal
  reconstruction,'' \emph{{IEEE} Trans. Inf. Theory}, vol.~55, no.~5, pp.
  2230--2249, May 2009.

\bibitem{candes2005decoding}
E.~J. Cand{\`e}s and T.~Tao, ``Decoding by linear programming,'' \emph{{IEEE}
  Trans. Inf. Theory}, vol.~51, no.~12, pp. 4203--4215, 2005.

\bibitem{candes2008restricted}
E.~J. Cand{\`e}s, ``The restricted isometry property and its implications for
  compressed sensing,'' \emph{Comptes Rendus Mathematique}, vol. 346, no.~9,
  pp. 589--592, 2008.

\bibitem{baraniuk2008simple}
R.~Baraniuk, M.~Davenport, R.~DeVore, and M.~Wakin, ``A simple proof of the
  restricted isometry property for random matrices,'' \emph{Constr. Approx.},
  vol.~28, no.~3, pp. 253--263, 2008.

\bibitem{Wed1972perturbation}
P.-{\AA}. Wedin, ``Perturbation bounds in connection with singular value
  decomposition,'' \emph{BIT Numer. Math.}, vol.~12, no.~1, pp. 99--111, 1972.

\end{thebibliography}

\end{document}